%% file: La_Co_Searching.tex
\documentclass[11pt,onecolumn]{IEEEtran}

\usepackage{graphicx}
\usepackage{epsfig}
\usepackage{amssymb}
\usepackage{amsbsy}
\usepackage{amsmath}
\usepackage{cite}
\usepackage{bbm}
\usepackage{subcaption}
\usepackage{cleveref}
\usepackage{url}
\usepackage{xcolor}
\usepackage{geometry}
\usepackage[export]{adjustbox}
\input{macros}

\renewcommand{\baselinestretch}{1.5}

\begin{document}
\title{Anomaly Search with Multiple Plays under Delay and Switching Costs}
\author{Tidhar Lambez and Kobi Cohen
\thanks{Tidhar Lambez and Kobi Cohen are with the School of Electrical and Computer Engineering, Ben-Gurion University of the Negev, Beer-Sheva 84105, Israel. Email: tidharl@post.bgu.ac.il, yakovsec@bgu.ac.il}
\thanks{This work has been submitted to the IEEE for possible publication. Copyright may be transferred without notice, after which this version may
no longer be accessible.}
\thanks{A short version of this paper was presented in the IEEE International Conference on Acoustics, Speech and Signal Processing (ICASSP) 2021 \cite{lambez2021searching_ICASSP}.}
\thanks{This research was supported by the ISRAEL SCIENCE FOUNDATION
(grant No. 2640/20), and by the Israeli National Cyber Bureau via the Cyber Security Research Center at Ben-Gurion University of the Negev.}
}

\maketitle

\begin{abstract}
\label{sec:abstract}
The problem of searching for $L$ anomalous processes among $M$ processes is considered. At each time, the decision maker can observe a subset of $K$ processes (i.e., multiple plays). The measurement drawn when observing a process follows one of two different distributions, depending on whether the process is normal or abnormal. The goal is to design a policy that minimizes the Bayes risk which balances between the sample complexity, detection errors, and the switching cost associated with switching across processes. We develop a policy, dubbed consecutive controlled sensing (CCS), to achieve this goal. On the one hand, by contrast to existing studies on controlled sensing, the CCS policy senses processes consecutively to reduce the switching cost. On the other hand, the policy controls the sensing operation in a closed-loop manner to switch between processes when necessary to guarantee reliable inference. We prove theoretically that CCS is asymptotically optimal in terms of minimizing the Bayes risk as the detection error approaches zero (i.e., the sample complexity increases). Simulation results demonstrate strong performance of CCS in the finite regime as well.\vspace{2mm}

Index Terms - \textnormal{Anomaly detection, controlled sensing, active hypothesis testing, sequential design of experiments.}
\end{abstract}

\section{Introduction}
\label{sec:introduction}

We consider the problem of searching for $L$ anomalous processes among $M$ processes (where $L\leq M$). Borrowing target search terminologies, we refer to the anomalous process as the target, which can be located in any of $M$ cells, whereas empty cells represent the normal processes. At each time, the decision maker can search for the target over $K$ cells ($1 \leq K \leq M$) (i.e., multiple plays). Each cell, at each time when selected, generates a random observation drawn from one of two different distributions $f$ and $g$, depending on whether the process is normal or abnormal. The observations are independent across time and cells.
Each time step incurs a cost $c$ (i.e., cost for delay) until decision is made. In addition, each time the decision maker observes a different set of $K$ processes incurs a switching cost proportional to the number of processes which are different from the previous play, as commonly assumed in online learning settings with switching cost (see \cite{agrawal1988asymptotically, agrawal1990multi} and subsequent studies). The objective is to design a search strategy that minimizes the Bayes risk which balances between the sample complexity, detection errors, and the switching cost associated with switching across processes.

The anomaly detection problem finds applications in intrusion detection in cyber systems for quickly locating compromised nodes by detecting statistical anomalies, spectrum scanning in cognitive radio networks for quickly locating idle channels for transmission, and event detection in sensor networks for quickly locating target events. A special case of a switching cost can be represented by an additional delay due to switching across processes in various data fusion and communication systems.

\subsection{Connection with Sequential Design of Experiments}

The anomaly detection problem considered in this paper has a connection with the classical sequential design of experiments problem (a.k.a. controlled sensing \cite{nitinawarat2013controlled} or active hypothesis testing \cite{naghshvar2013active} in recent studies) first studied by Chernoff \cite{chernoff1959sequential}. Compared with the classical sequential hypothesis testing pioneered by Wald~\cite{Wald_1947_Sequential}, where the observation model under each hypothesis is predetermined, the sequential design of experiments has a control aspect that allows the decision maker to choose the experiment to be conducted at each time. Intuitively, as more observations are collected, the decision maker becomes more certain about the true hypothesis, which in turn leads to better choices of experiments. Chernoff has established a \emph{randomized} strategy, referred to as the Chernoff test which is asymptotically optimal in terms of minimizing the detection delay as the error probability diminishes. Variations and extensions of the problem and the Chernoff test were studied and suggested recently in~\cite{nitinawarat2013controlled, nitinawarat2012controlled, nitinawarat2015controlled, liu2015adaptive,  naghshvar2013sequentiality, naghshvar2013active, cohen2015active, kaspi2017searching,  huang2018active, tsopelakos2019sequential, tsopelakos2020sequential, kartik2019fixed, kartik2020testing, gurevich2019sequential, hemo2020searching, gafni2021federated}. When translating the sequential experimental design setting to the realm of the anomaly detection problem considered here, the experiments represent the set of $K$ cells that can be probed at each time. In particular, under each hypothesis that the target is located in a particular cell, the distribution (either $f$ or $g$) of the next observation depends on the cell chosen to be searched. The Chernoff test and its variations~\cite{chernoff1959sequential, nitinawarat2013controlled, nitinawarat2015controlled, naghshvar2013active, cohen2015active, huang2018active, tsopelakos2019sequential, kartik2019fixed} apply to our problem. However, switching cost was not considered in all these studies, which results in linear order of the accumulated switching cost with time.

There are a few recent studies on sequential experimental design that considered the minimization of switching cost. Specifically, in~\cite{vaidhiyan2015active, vaidhiyan2017neural}, the Chernoff test was modified to penalize for switching of actions, dubbed Sluggish Procedure A. At each decision instance, the next action is determined according to the current posterior probability and the previous action, depending on a Bernoulli random variable. The policy was shown to be asymptotically optimal, but the accumulated switching cost was unbounded with time. Unlike~\cite{vaidhiyan2015active, vaidhiyan2017neural}, here we show that deterministic policy achieves asymptotic optimality with bounded switching cost. In addition to the theoretical asymptotic improvement, we demonstrate significant performance gains in the finite regime as well. In \cite{chen2019active}, the authors developed a deterministic strategy with a bounded switching cost under the anomaly detection setting considered here, but only for the case where a single cell is observed at a time ($K=1$). In this paper, we consider the multiple play case ($K\geq 1$), which adds significant challenges in the algorithm design and performance analysis.

\subsection{Main Results}

Sequential detection problems involving multiple processes are partially-observed Markov decision processes (POMDP) \cite{Castanon_1995_Optimal} which have exponential computational complexity in general. As a result, computing optimal search policies is intractable (except for some special cases of observation distributions as in~\cite{zigangirov1966problem, Castanon_1995_Optimal}). For tractability, a commonly adopted performance measure is asymptotic optimality in terms of minimizing the sample complexity as the error probability approaches zero (see, for example, classic and recent results as overviewed in the previous subsection). The focus of this paper is thus on asymptotically optimal strategies with low computational complexity. We develop a deterministic strategy with multiple plays, dubbed Consecutive Controlled Sensing (CCS), to solve the anomaly detection problem with low-complexity implementations. Specifically, CCS consists of exploration and exploitation phases. During exploration, the cells are probed in a round-robin manner for quickly inferring the most likely cells that contain the target. During exploitation phases, the most informative observations are collected consecutively based on the current belief on the location of the targets. The computational complexity is only linear with the number of cells and time steps (due to updating the beliefs and probing order). The structure of the CCS policy allows to sense processes consecutively in exploitation phases to reduce the switching cost, while the exploration time is shown to be bounded. This is in sharp contrast with the linear order of switching cost under existing search strategies~\cite{chernoff1959sequential, nitinawarat2013controlled, nitinawarat2015controlled, naghshvar2013active, cohen2015active, huang2018active}. We analyze the performance of the CCS policy theoretically, and show that it is asymptotically optimal in terms of minimizing the total cost as the error probability approaches zero. Simulation results demonstrate significant performance gains over existing policies in the finite regime as well.

\subsection{Other Related Work}

Optimal solutions for target search or target whereabout problems have been obtained under some special cases when a single location is probed at a time (i.e., single play). Modern application areas of search problems with limited sensing resources include narrowband spectrum scanning \cite{egan2017fast}, event detection by a fusion center that communicates with sensors using narrowband transmissions~\cite{blum2008energy,cohen2011energy, sery2020analog}, and sensor visual search studied recently by neuroscientists \cite{vaidhiyan2015learning}.
Results under the sequential setting with observation control can be found in~\cite{zigangirov1966problem, Klimko_1975_Optimal, Dragalin_1996_Simple, Stone_1971_Optimal, blum2008energy, cohen2011energy, banerjee2012data, cohen2014optimal}. Specifically, optimal policies were derived in~\cite{zigangirov1966problem, Klimko_1975_Optimal, Dragalin_1996_Simple} for the problem of quickest search over Wiener processes. In~\cite{Stone_1971_Optimal, cohen2014optimal}, optimal search strategies were established under the constraint that switching to a new process is allowed only when the state of the currently probed process is declared. Optimal policies under general distributions and unconstrained search model remain an open question. In this paper we address this question under the asymptotic regime as the error probability approaches zero. Optimal search strategies when a single location is probed at a time and a fixed sample size have been established under binary-valued measurements \cite{Tognetti_1968_An, Kadane_1971_Optimal}, and under known symmetric distributions of continuous observations \cite{Castanon_1995_Optimal}. In this paper, however, we focus on the sequential setting.

Sequential tests for hypothesis testing problems have attracted much attention since Wald's pioneering work on sequential analysis \cite{Wald_1947_Sequential}. The reason for this is due to the sequential tests property of reaching a decision at a much earlier stage than would be possible with fixed-size tests. Wald established the Sequential Probability Ratio Test (SPRT) for a binary hypothesis testing of a single process. Under the simple hypothesis case, the SPRT is optimal in terms of minimizing the expected sample size under given type \emph{$I$} and type \emph{$II$} error probability constraints. For a single process, various extensions for M-ary hypothesis testing and testing of composite hypotheses were studied in~\cite{schwarz1962asymptotic,lai1988nearly, pavlov1991sequential,tartakovsky2002efficient,draglia1999multihypothesis,heydari2018controlled, deshmukh2019sequential}. In these cases, asymptotically optimal performance can be obtained as the error probability approaches zero. In this paper, we focus on an asymptotically optimal strategy for sequential search of a target over multiple processes under constraints on the probing capacity. By contrast, searching for targets without constraints on the probing capacity, where all processes are probed at each given time, were considered by different models in~\cite{song2017asymptotically,Dragalin_1996_Simple,tartakovsky2002efficient,song2019sequential}.

Another set of related studies considered sequential detection over multiple independent processes~\cite{li2009restless, Lai_2011_Quickest, Malloy_2012_Quickest, Tajer_2013_Quick, Caromi_2013_Fast, malloy2014sequential, cohen2015asymptotically, heydari2016quickest, heydari2017quickest, heydari2018quickest}. In particular, in~\cite{Lai_2011_Quickest}, the problem of identifying the first abnormal sequence among an infinite number of i.i.d. sequences was considered. An optimal cumulative sum (CUSUM) test has been established under this setting. Further studies on this model can be found in~\cite{Malloy_2012_Quickest, Tajer_2013_Quick, malloy2014sequential}. While the objective of finding rare events or a single target is considered in~\cite{Lai_2011_Quickest, Malloy_2012_Quickest, Tajer_2013_Quick, malloy2014sequential} is similar to that of this paper, the main difference is that in~\cite{Lai_2011_Quickest, Malloy_2012_Quickest, Tajer_2013_Quick, malloy2014sequential} the search is done over an infinite number of i.i.d processes, where the state of each process (normal or abnormal) is independent of other processes. This results in open-loop search strategies, which is fundamentally different from the setting in this paper. Other recent studies include searching for correlation structures of Markov networks \cite{heydari2016quickest}, searching for a moving Markovian target \cite{leahy2016always}, and approximating optimal policies via deep reinforcement learning techniques\cite{zhong2019deep, livne2020pops}.

\section{System Model and Problem Statement}
\label{sec:problem}

We consider the problem of sequentially detecting $L$ anomalous processes among $M$ processes. We formalize the problem for the case where $L=1$, and discuss the extension for multiple targets detection ($L\geq 1$) in Subsection \ref{ssec:multipleTargets_development}. The anomalous process is referred to as the target, which can be located in any of $M$ cells, whereas empty cells represent the normal processes. If the target is in cell $m$, we say that hypothesis $H_m$ is true, and the a priori probability that $H_m$ is true is denoted by $\pi_m$, where $\sum_{m=1}^{M} \pi_m =1$. 

At each time, $K$ cells can be probed simultaneously by the decision maker (equipped with $K$ probing machines, where $1\leq K\leq M$). When cell $m$ is probed at time $n$, an observation $y_m(n)$ is drawn independently from a known distribution. If hypothesis $H_m$ is true, $y_m(n)$ follows distribution $g(y)$. Otherwise, if hypothesis $H_m$ is false, $y_m(n)$ follows distribution $f(y)$. Let $\mathbf{P}_m$ be the probability measure under hypothesis $H_m$ and $\mathbf{E}_m$ be the operator of expectation with respect to the measure $\mathbf{P}_m$.

Let $\Gamma=(\tau,\delta,\phi)$ be a sequential anomaly detection policy, where $\tau$ is the stopping time when the decision maker finalizes the search and declares the location of the target. 
Let $\delta\in\{1,...,M\}$ be a decision rule, where $\delta=m$ if the decision maker declares that $H_m$ is true, and $\phi=\{\phi(n)\}_{n\geq1}$ is the sequence of selection rules, where $\phi(n)$ is a $K$-sized permutation from the set $\left\{1, ..., M\right\}$, indicating which $K$ cells are probed at time $n$. Let $\mathbf{P}_e(\Gamma)= \sum_{m=1}^{M} \pi_m \alpha_m(\Gamma)$ be the probability of error under policy $\Gamma$, where $\alpha_m(\Gamma)=\mathbf{P}_m(\delta \neq m;\Gamma)$ is the probability of declaring $\delta \neq m$ when $H_m$ is true. 
Let $\mathbf{E}(\tau|\Gamma)=\sum_{m=1}^{M} \pi_m \mathbf{E}_m(\tau|\Gamma)$ be the expected detection delay under policy $\Gamma$.

We adopt a Bayesian approach by assigning a cost of $c$ for each observation, a cost of $s$ for each switch, where it is assumed that $s=O(c)$, and a loss of 1 for a wrong declaration by time $t=\tau$. Specifically, each time step $t<\tau$ incurs a cost $c+N_s(t)\cdot s$, where $N_s(t)$ is the number of processes which are different from the previous play\footnote{This definition of switching cost is commonly used in online learning settings (see \cite{agrawal1988asymptotically, agrawal1990multi} and subsequent studies). Nevertheless, the theoretical analysis applies to all cases where a cost $s=O(c)$ is added at each time step in which switching occurs, regardless of the exact number of processes which are different from the previous play.}. Let $\tau_{switch}\triangleq\sum_{t=1}^{\tau}N_s(t)$ be the total number of switchings at the time of stopping, and $\mathbf{E}(\tau_{switch}|\Gamma)=\sum_{m=1}^{M} \pi_m \mathbf{E}_m(\tau_{switch}|\Gamma)$ be its expected value under policy $\Gamma$. The Bayes risk under policy $\Gamma$, which balances between detection errors, sample complexity, and switching cost is defined by \cite{chen2019active}:
\beq
    R_m(\Gamma)\triangleq\alpha_m(\Gamma) + c\mathbf{E}_m(\tau|\Gamma)+ s\mathbf{E}_m(\tau_{switch}|\Gamma),
\eeq
when hypothesis $H_m$ is true. The average Bayes risk under policy $\Gamma$ is thus given by:
\beq
    \label{eq:bayes_risk}
    R(\Gamma)= \sum\limits_{m=1}^{M} \pi_m R_m(\Gamma)= \mathbf{P}_e(\Gamma) + c\mathbf{E}(\tau|\Gamma)+ s\mathbf{E}(\tau_{switch}|\Gamma).
\eeq
The objective is to find a policy $\Gamma$ that minimizes the Bayes risk $R(\Gamma)$:
\beq
\label{eq:objective_risk}
    \underset{\Gamma}{\inf} \; R(\Gamma).
\eeq
In the next section we develop the CCS policy to achieve this goal in the asymptotic limit as $c\rightarrow 0$.

\begin{figure*}[t]
	\centering \epsfig{file=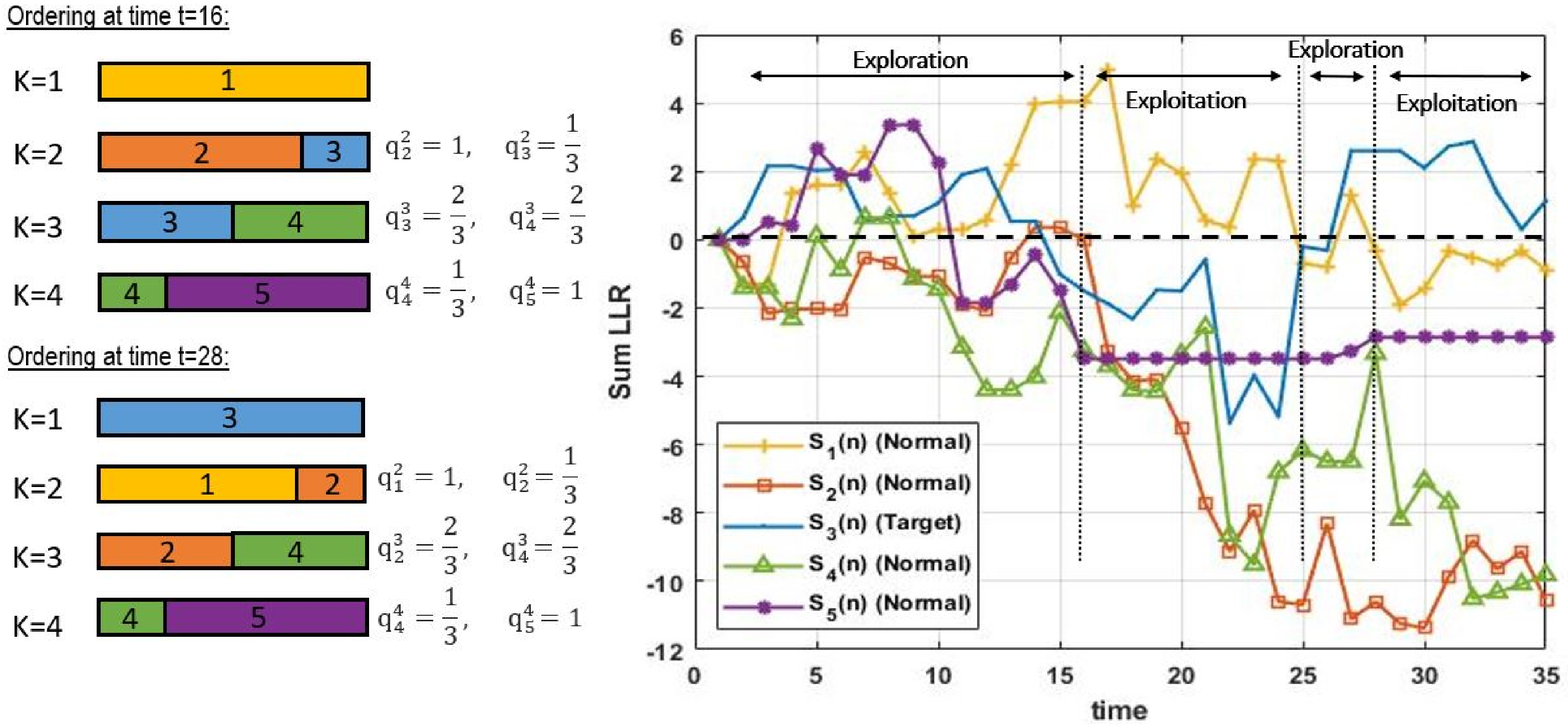,
		width=1.05\textwidth,height=0.26\textheight, center}
	\caption{\small {An illustration of the CCS policy for Case $1$ when $K=4$, $M=5$, and cell $3$ contains the target. At times $0\leq t\leq 16$, and $25\leq t\leq 28$, $N_{H_1}(n)\neq 1$, and thus exploration phases are conducted. After the first exploration phase ends, cell $1$ is suspected to be the target (mistakenly) and ordered in probing machine $k=1$. Then, an exploitation phase starts at time $t=17$. After the second exploration phase ends, cell $3$ is suspected to be the target (correctly) and ordered in probing machine $k=1$. Then, an exploitation phase starts at time $t=29$. We show in the analysis that the time spent during exploration phases is bounded. In each probing machine $k\neq 1$, two processes are ordered such that their portions are summed to $\frac{M-1}{K-1}=\frac{4}{3}$.}}
	\label{fig:ordering}
\end{figure*}

\section{The Consecutive Controlled Sensing (CCS) policy}
\label{sec:CCSpolicy}

In this section we develop a deterministic policy, dubbed Consecutive Controlled Sensing (CCS) policy, to solve (\ref{eq:objective_risk}). For ease of presentation and explanation, we first present the case of a single anomaly, where $L=1$, in Subsection \ref{ssec:CCS_development}, and then the more general case of multiple anomalies, where $L\geq 1$, in Subsection \ref{ssec:multipleTargets_development}. The asymptotic optimality of CCS is analyzed in Theorem \ref{th:asymptotic_optimality} in Subsection \ref{ssec:performance_CCS}. \vspace{0.3cm}

We start by introducing notations that we use in the development of the CCS policy. 
Let 
\beq
    \ell_m(n) \triangleq \log{\bigg(\frac{g(y_m(n))}{f(y_m(n))}\bigg)},
\eeq
and
\beq
    S_m(n) \triangleq \sum\limits_{t=1}^{n} \ell_m(t) \textbf{1}_m(t)
\eeq
be the log-likelihood ratio (LLR) and the observed sum LLR of cell $m$ at time $n$, respectively, where $\textbf{1}_m(n)$ is the indicator function on whether cell $m$ is observed at time $n$. $\textbf{1}_m(n)=1$ if cell $m$ is observed at time $n$, and $\textbf{1}_m(n)=0$ otherwise.

We define $m^{(i)} (n)$ as the index of the cell with the $i$th highest observed sum LLR at time $n$. Let
\beq
    \Delta S(n) \triangleq S_{m^{(1)} (n)}(n) - S_{m^{(2)} (n)}(n)
    \label{eq:deltaSn}
\eeq
be the difference between the highest and the second highest observed sum LLR at time $n$. 
Let $\mathcal{H}_1(n)=\left\{m\;:\; S_m(n)>0\right\}$ be the set of cells whose sum LLRs are greater than zero at time $n$ with cardinality $|\mathcal{H}_1(n)|=N_{H_1}(n)$. Let $D(g||f)$ denote the Kullback-Leibler (KL) divergence between two distributions $g$ and $f$, indicating the expected LLR between distributions $g$ and $f$, where the expectation is taken with respect to $g$, $D(g||f)=\mathbf{E}_{\sim g}[\log(g/f)]$. Similarly, the KL divergence between distributions $f$ and $g$ is the minus expected LLR between distributions $g$ and $f$, where the expectation is taken with respect to $f$, $D(f||g)=\mathbf{E}_{\sim f}[\log(f/g)]=-\mathbf{E}_{\sim f}[\log(g/f)]$. As a result, the sum LLR, generated by sampling the target, is a random walk with a positive increment with rate $D(g||f)$, while the sum LLR, generated by sampling a normal process, is a random walk with a negative increment with rate $D(f||g)$.

In the design of CCS we define two cases, depending on the order of the system size $M$ and $L\cdot (D(f||g)/D(g||f)+1)$. In Case 1, $L\cdot (D(f||g)/D(g||f)+1) \leq M$ holds, where in Case 2, $L\cdot (D(f||g)/D(g||f)+1) > M$ holds. Define
\begin{equation}
\label{eq:rate_function}
\bea{l}
    I^*(M,K,L) \vspace{2mm}\\
    \triangleq \begin{cases}
        D(g||f)+\frac{(K-L)D(f||g)}{M-L}, &\text{if $L\cdot \bigg(\frac{D(f||g)}{D(g||f)}+1 \bigg) \leq M$} \vspace{4mm} \\ 
        \frac{K \cdot D(f||g)}{M-L}, &\text{if $L\cdot \bigg(\frac{D(f||g)}{D(g||f)}+1 \bigg) > M$}
    \end{cases}
\ena
\end{equation}
which will be shown later to be the rate function of the asymptotically optimal performance. The intuition behind the rate function can be easily understood by the case of a single target ($L=1$), and noting that $D(f||g)/D(g||f)+1\leq M$ is equivalent to $D(f||g)/(M-1)\leq D(g||f)$. In Case $1$ (i.e., $D(f||g)/D(g||f)+1\leq M$), the algorithm should always decide to sample the (suspected) target since it contributes to increasing $\Delta S(n)$ with rate $D(g||f)$. Then, the remaining $K-1$ probing machines should be used to sample the rest $M-1$ (suspected) normal processes (which contributes to increasing $\Delta S(n)$ with rate $(K-1)D(f||g)/(M-1)$). This results in a total rate of $I^*(M,K,L)=D(g||f)+\frac{(K-1)D(f||g)}{M-1}$. In Case 2 (i.e., $D(f||g)/D(g||f)+1 >M$), intuitively, the algorithm should not sample the (suspected) target since sampling each normal process is preferred as it contributes to increasing $\Delta S(n)$ with rate $D(f||g)/(M-1)$. This results in a total rate of $I^*(M,K,L)=\frac{K\cdot D(f||g)}{M-1}$. It should be noted that the algorithm does not know which process is target or normal. The CCS algorithm infers the true state online and adapts to the desired strategy required to minimize the Bayes risk. In the following subsections we present the CCS algorithm in each of the two cases separately.

\subsection{Detection of a single anomaly}
\label{ssec:CCS_development}

\subsubsection{Case $1$: $D(f||g)/D(g||f)+1 \leq M$}
\label{sssec:CCS_development_case1}

The CCS policy has a structure of exploration and exploitation epochs. We describe CCS policy below for a single target with respect to time index $n$. An illustration of the algorithm is presented in Fig. \ref{fig:ordering}.
\vspace{0.2cm}
 
\noindent
\textbf{1) (Exploration phase:)} If $N_{H_1}(n)\neq 1$, then probe the cells one by one in a round-robin manner, while updating the sum LLR of each cell, and go to Step $1$ again. Otherwise, go to Step $2$.\\
For example, in Fig. \ref{fig:ordering}, exploration phases occur at times $0\leq t\leq 16$, and $25\leq t\leq 28$.\vspace{0.2cm}

\noindent
\textbf{2) (Ordering the $M$ processes based on current beliefs:)} Set the probing order before the exploitation phase starts. The processes are divided into $K$ probing machines, where process $m^{(1)} (n)$ is in probing machine $k=1$ and all other $M-1$ processes are ordered into probing machines\footnote{Note that when $K=1$, only cell $m^{(1)} (n)$ is ordered in machine $k=1$, and the total estimated sensing time is set to $-\log c/D(g||f)$. Then, go to Step $3$ and take observations until $\Delta S(n)$ is greater than $-\log c$. No suspected normal cells are probed in this case.} $k=2, ..., K$. The ordering is done by first concatenating all $M-1$ suspected normal cells one after another in an arbitrary order. The length of each process is determined by the estimated sensing time, $-\frac{(K-1)\log c}{(M-1)I^*(M,K,L)}$ for each process (i.e., a total length of $-\frac{(K-1)\log c}{I^*(M,K,L)}$). Then, the $M-1$ cells are equally divided to $K-1$ parallel machines, where the total estimated sensing time for all processes in each machine is thus\footnote{In the analysis we show that this allocation guarantees asymptotic optimality of the test.} $-\log c/I^*(M,K,L)$. Note that the actual sensing time is random and the estimated sensing time is only used to determine the allocation of processes in the probing machines based on the expected asymptotic time required to sample each process until the termination of the test. Let $\mathcal{H}^{k}(t_i)\subseteq\{1,...,M\}$ be the set of indices, indicating which cells are ordered for probing by machine $k$ during exploitation phase $t_i$. For convenience, when there is no ambiguities, we often write $\mathcal{H}^{k}$ to denote this set. For machine $k\neq 1$, note that the first and last processes might be split and allocated on two probing machines at most (e.g., processes $2 \mbox{ and } 4$ at ordering time $t=28$ in Fig. \ref{fig:ordering}). All other processes in the probing machine are not split. Denote $\mathcal{K}_{j}$ as the set of machines that probe process $j$. We define $q_j^{k}$ as the fraction of the estimated sensing time in which process $j$ is sampled in probing machine $k\neq 1$, where $q_j^{k}=0$ if $k\nin\mathcal{K}_{j}$, $q_j^{k}=1$ if $k\in\mathcal{K}_{j}$ and process $j$ is sampled by a single machine (i.e., $|\mathcal{K}_{j}|=1$), and $0<q_j^{k}<1$, $0<q_j^{k+1}=1-q_j^{k}<1$ when process $j$ is split between two successive machines $k, k+1\in\mathcal{K}_{j}$ (i.e., $|\mathcal{K}_{j}|=2$). Note that by these definitions, we get $\sum_k q_j^{k}=1$. By the ordering mechanism, for each probing machine $k$, it holds that summing $q_j^{k}$'s over all cells probed by machine $k\neq 1$ yields:
\begin{equation}
    \sum\limits_{j\in\mathcal{H}^{k}} q_j^{k} = \frac{M-1}{K-1},
\end{equation}
as $M-1$ suspected normal cells are allocated in probing machines $k=2,..,K$. 
Note that a process could be allocated in two probing machines, but cannot be observed simultaneously (we discuss the mechanism that guarantees the feasibility of the sensing process in such cases later). Go to Step $3$.\\
An illustration of the ordering procedure is presented in Fig. \ref{fig:ordering} at times $t=16$ and $t=28$. \vspace{0.2cm}

\noindent
\textbf{3) (Exploitation phase:)} Whenever $N_{H_1}(n)\neq 1$ during this step go back to Step $1$. As long as $N_{H_1}(n)= 1$, probe cells $\phi(n)$ consecutively following the order in Step $2$, and update the sum LLRs accordingly. For each process, say $j'$, which is ordered in only one of probing machines $k=2, ..., K$ (which is likely to be normal), take observations consecutively until $S_{j'}(n)$ is smaller than
\begin{equation}
    \frac{(K-1) D(f||g) \log c}{(M-1) I^*(M,K,L)}.
\end{equation}
Then, sample the next process following the order in Step $2$. If the process, say $j$, is ordered in two of probing machines $k \mbox{ and } k+1$, where $k\neq 1$ (e.g., processes $2 \mbox{ and } 4$ in ordering time $t=28$ in Fig. \ref{fig:ordering}), then start by taking the first fraction of the set of observations in machine $k+1$ as divided by the scheduler, referred to as $q_j^{k+1}$, such that $S_j(n)$ is smaller than\vspace{0.1cm}
\begin{equation}
    q_j^{k+1} \cdot \frac{(K-1) D(f||g) \log c}{(M-1) I^*(M,K,L)}.
\end{equation}
When sampling process $j$ in the second probing machine $k$ later, take observations until $S_j(n)$ is smaller than\vspace{0.1cm}
\begin{equation}
    \small{(q_j^{k} + q_j^{k+1})\cdot \frac{(K-1) D(f||g) \log c}{(M-1) I^*(M,K,L)} = \frac{(K-1) D(f||g) \log c}{(M-1) I^*(M,K,L)}}.\vspace{2mm}
\end{equation} 
For the process in probing machine $k=1$ (which is likely to be the target, for instance, process $3$ in Fig. \ref{fig:ordering} at time $t=28$), take observations consecutively until $\Delta S(n)$ is greater than $-\log c$. The completion of this step determines the stopping time of the test $\tau$. The decision rule is given by $\delta = m^{(1)}(\tau)$.

Note that a process cannot be probed simultaneously in more than one probing machine. Therefore, in a case of a conflict, an additional idle time is added to prevent this situation and guarantee the feasibility of the sensing process. This applies only to processes suspected to be normal (i.e., allocated in machines $k=2,...,K$), as only those processes are split between probing machines according to the ordering mechanism. An illustration of this scenario is presented in Fig. \ref{fig:idled_delay}.

\begin{figure}[ht]
    \centering
    \includegraphics[width=0.2\textwidth]{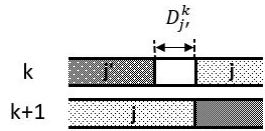}
    \caption{\small{A scenario when additional idle time is added to normal process $j'$, referred as $D_{j'}^{k}$, to guarantee that process $j$ is not sampled by both machines $k$ and $k+1$ simultaneously.}}
    \label{fig:idled_delay}
\end{figure}

\subsubsection{Case 2: $D(f||g)/D(g||f)+1 > M$}
\label{sssec:CCS_development_case2}

\begin{figure}[t!]
	\centering \epsfig{file=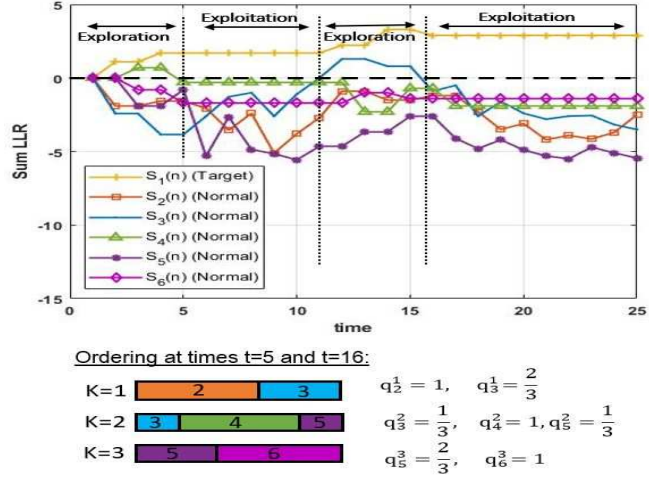,		width=0.5\textwidth,height=0.25\textheight, center}
	\caption{\small{An illustration of the CCS policy for Case $2$ when $K=3$, $M=6$, and cell $1$ contains the target. At times $0\leq t\leq 5$, and $11\leq t\leq 16$, $N_{H_1}(n)\neq 1$, and thus exploration phases are conducted. After the first exploration phase ends, cell $1$ is suspected to be the target (correctly) and not ordered to be sampled. Then, an exploitation phase starts at time $t=6$. After the second exploration phase ends, cell $1$ again is suspected to be the target and not ordered to be sampled. Then, an exploitation phase starts at time $t=17$.}}
	\label{fig:ordering_case2}
\end{figure}

The CCS policy has a structure of exploration and exploitation epochs again, where similar to Case $1$, the purpose of the exploration phases is to determine the location of the target cell. We describe CCS below with respect to time index $n$. An illustration of the algorithm is presented in Fig. \ref{fig:ordering_case2}.
In this case, only the suspected normal cells are ordered for probing by $K$ probing machines at a time, as intuitively explained earlier.
\vspace{0.2cm}
 
\noindent
\textbf{1) (Exploration phase:)} Implement the same exploration phase as in Case $1$, described in subsection \ref{sssec:CCS_development_case1}.\\
For example, in Fig. \ref{fig:ordering_case2}, exploration phases occur at times $0\leq t\leq 5$, and $11\leq t\leq 16$.\vspace{0.2cm}

\noindent 
\textbf{2) (Ordering the $M-1$ processes based on current beliefs:)} Set the probing order before the exploitation phase starts. The processes are divided into $K$ probing machines, where all suspected $M-1$ normal cells are ordered for probing by machines $k=1,...,K$ (process $m^{(1)} (n)$ is not probed by any machine). The ordering is done by first concatenating all $M-1$ suspected normal cells one after another in an arbitrary order. The length of each process is determined by the estimated sensing time, $-\frac{K\log c}{(M-1)I^*(M,K,L)}$ for each process (i.e., a total length of $-\frac{K\log c}{I^*(M,K,L)}$). Then, the $M-1$ cells are equally divided to $K$ parallel machines, where the total estimated sensing time for all processes in each machine is thus $-\log c/I^*(M,K,L)$. Since each normal cell is scheduled in probing machine $k$, we use the definition of $q_j^{k}$ for all $k$. By the ordering mechanism, for each probing machine $k$, it holds that summing $q_j^{k}$'s over all cells probed by machine $k$ yields:
\begin{equation}
    \sum\limits_{j\in\mathcal{H}^{k}} q_j^{k} = \frac{M-1}{K},
\end{equation}
as $M-1$ suspected normal cells are allocated in $K$ probing machines. Go to Step $3$.\\
An illustration of the ordering procedure is presented in Fig. \ref{fig:ordering_case2} at times $t=5$ and $t=16$. \vspace{0.2cm}

\noindent
\textbf{3) (Exploitation phase:)} Whenever $N_{H_1}(n)\neq 1$ during this step go back to Step $1$. As long as $N_{H_1}(n)= 1$, probe cells $\phi(n)$ consecutively following the order in Step $2$, and update the sum LLRs accordingly. For each process, say $j'$, which is ordered in only one of probing machines $k=1, ..., K$, take observations consecutively until $S_{j'}(n)$ is smaller than
\begin{equation}
    \frac{K\cdot D(f||g) \log c}{(M-1) I^*(M,K,L)}.
\end{equation}
Then, sample the next process following the order in Step $2$. If the process, say $j$, is ordered in two of probing machines $k \mbox{ and } k+1$ (e.g., processes $3 \mbox{ and } 5$ in ordering time $t=5$ in Fig. \ref{fig:ordering_case2}), then start by taking the first fraction of the set of observations in machine $k+1$ as divided by the scheduler, referred to as $q_j^{k+1}$, such that $S_j(n)$ is smaller than\vspace{0.1cm}
\begin{equation}
    q_j^{k+1} \cdot \frac{K \cdot D(f||g) \log c}{(M-1) I^*(M,K,L)}.
\end{equation}
When sampling process $j$ in the second probing machine $k$ later, take observations until $S_j(n)$ is smaller than\vspace{0.1cm}
\begin{equation}
    \small{(q_j^{k} + q_j^{k+1})\cdot \frac{K\cdot D(f||g) \log c}{(M-1) I^*(M,K,L)} = \frac{K\cdot D(f||g) \log c}{(M-1) I^*(M,K,L)}}.
\end{equation}
Continue to take observations consecutively until $\Delta S(n)$ is greater than $-\log c$. Note that the additional idle times are added only to processes suspected to be normal (i.e., allocated in all $K$ machines) to guarantee the feasibility of the sensing process. The completion of this step determines the stopping time of the test $\tau$. The decision rule is given by $\delta = m^{(1)}(\tau)$. \vspace{2mm}

The intuition of the design of CCS policy is to exploit observations from the target and the normal processes to infer the location of the target. Intuitively speaking, when the sum LLRs of normal processes decrease, it is more likely that the process with the positive sum LLR is the target. We show in the analysis that the judicious design of the CCS policy achieves asymptotically optimal performance, with a bounded accumulated switching cost.

\begin{figure*}[t]
	\centering \epsfig{file=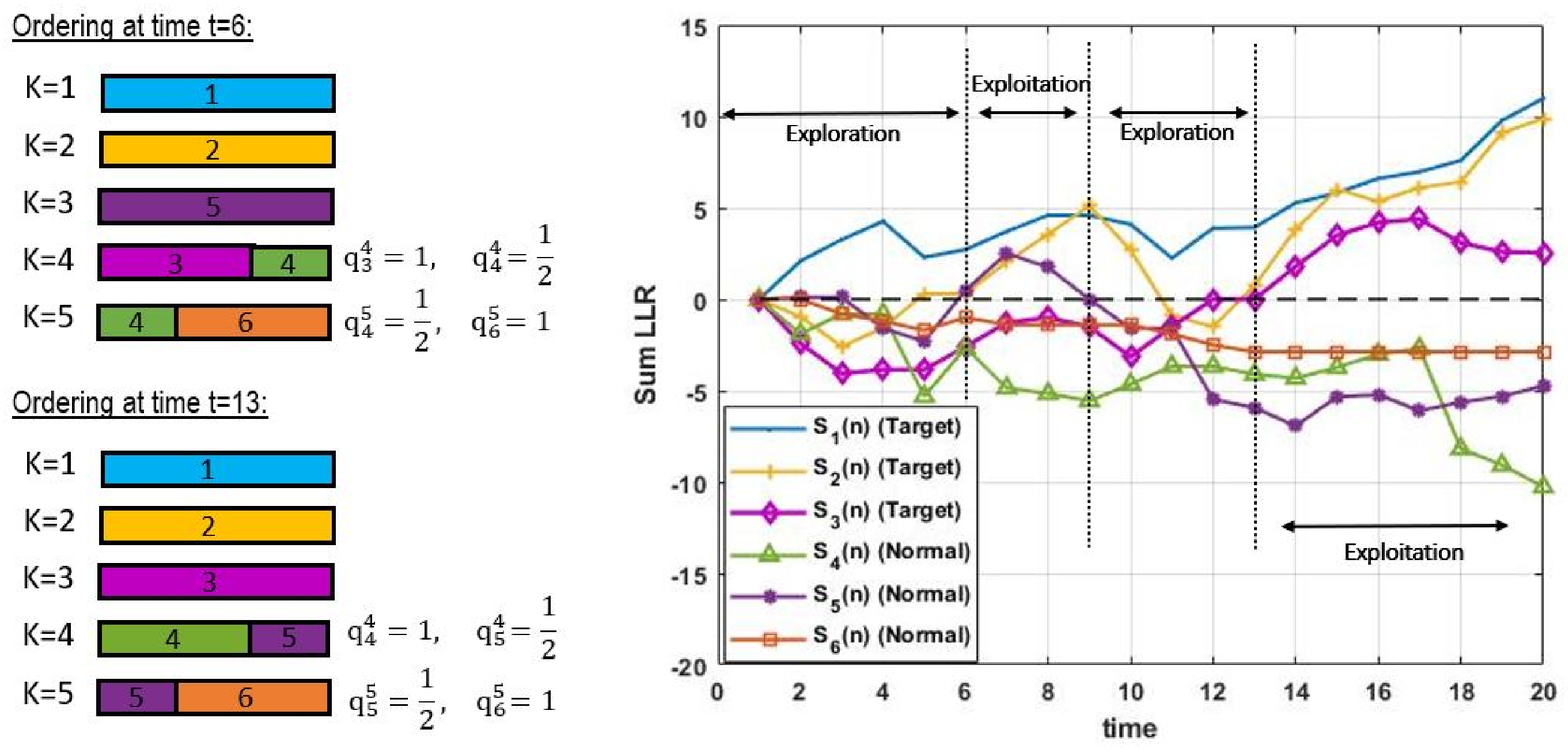,
		width=1.05\textwidth,height=0.26\textheight, center}
	\caption{\small {An illustration of the CCS policy when $L=3, K=5$, $M=6$, and cells $1,2 \mbox{ and } 3$ contain targets. At times $0\leq t\leq 6$ and $9\leq t\leq 13$, $N_{H_1}(n)\neq L$, and thus exploration phases are conducted. After the first exploration phase ends, cells $1,2 \mbox{ and 5}$ are suspected to be the targets (mistakenly) and ordered in probing machines $k=1,2,3$. Then, an exploitation phase starts at time $t=7$. After the second exploration phase ends, cells $1,2 \mbox{ and } 3$ are suspected to be the targets (correctly) and ordered in probing machines $k=1,2,3$. Then, an exploitation phase starts. In each probing machine $k=4,5$, two processes are ordered such that their portions are summed to $\frac{M-L}{K-L}=\frac{3}{2}$.}}
	\label{fig:ordering_multipleTar}
\end{figure*}

\subsection{Extension to Detection of Multiple Anomalies}
\label{ssec:multipleTargets_development}

In this subsection we extend the results reported in the previous subsection and analyze the case of $L\geq 1$ targets, where $L$ is known, and $K\geq L$.
Similar to the single-target case, the rate function in \eqref{eq:rate_function} can be intuitively explained as follows. In Case 1 (i.e., $D(f||g)/(M-L)\leq D(g||f)/L$), the algorithm should always sample the $L$ (suspected) targets, and the rest $K-L$ probing machines should be used to sample the rest $M-L$ (suspected) normal processes. This results in a total rate of $I^*(M,K,L)=D(g||f)+\frac{(K-L)D(f||g)}{M-L}$. In Case 2 (i.e., $D(f||g)/(M-L)>D(g||f)/L$), intuitively, the algorithm should not sample the (suspected) targets since sampling each normal process is preferred as it contributes to increasing the difference between the $L$th highest and the $(L+1)$th highest observed sum LLR with rate $D(f||g)/(M-L)$. This results in a total rate of $I^*(M,K,L)=\frac{K D(f||g)}{M-L}$. Next, we focus on presenting CCS for the case of a large-scale systems, where $M$ is sufficiently large, so that $L\cdot (D(f||g) /D(g||f) + 1) \leq M$ holds (i.e., Case 1). The design for Case 2 can be easily extended as discussed for the single-target case in Subsection \ref{sssec:CCS_development_case2}.
Let 
\beq
    \Delta_L S(n) \triangleq S_{m^{(L)} (n)}(n) - S_{m^{(L+1)} (n)}(n),
\eeq
be the difference between the $L$th highest and the $(L+1)$th highest observed sum LLR at time $n$. 
The decision rule declares a set of $L$ target locations $\delta_L = ( m^{(1)}(\tau), m^{(2)}(\tau), ..., m^{(L)}(\tau))$, where $\tau = \inf \{n:\Delta_L S(n)\geq -\log(c) \}$ is the stopping rule. The number of hypotheses in this case is $\big(\begin{smallmatrix}
  M\\
  L
\end{smallmatrix}\big)$. Let $\mathcal{D}\subset \{1,2,...,M\}$ denote a subset of $L$ cells with cardinality $|\mathcal{D}|=L$ indicating the location of all target cells, where the set of hypotheses is defined by $H_{\mathcal{D}}$. \vspace{0.5cm}

The CCS policy has a structure of exploration and exploitation epochs. We describe the policy below for multiple targets with respect to time index $n$. An illustration of the algorithm for the multi-target case is presented in Fig. \ref{fig:ordering_multipleTar}. \vspace{0.2cm}

\noindent
\textbf{1) (Exploration phase:)} If $N_{H_1}(n)\neq L$, then probe the cells one by one in a round-robin manner, while updating the sum LLR of each cell, and go to Step $1$ again. Otherwise, go to Step $2$.\vspace{0.2cm}

\noindent
\textbf{2) (Ordering the $M$ processes based on current beliefs:)} Set the probing order before the exploitation phase starts. The processes are divided into $K$ probing machines, where processes ${m^{(1)} (n),m^{(2)} (n), ..., m^{(L)} (n)}$ are in probing machines $k=1, ..., L$ correspondingly (note that $K\geq L$), and all other $M-L$ processes are ordered into probing machines\footnote{Note that when $K=L$, only cells ${m^{(1)} (n),m^{(2)} (n), ..., m^{(L)} (n)}$ are ordered in machines $k=1, ...,L$, and the total estimated sensing time is set to $-\log c/D(g||f)$. Then, go to Step $3$ and sample until $\Delta_L S(n)$ is greater than $-\log c$. No suspected normal cells are probed in this case.} $k=L+1, ..., K$. The ordering is done by first concatenating all $M-L$ suspected normal cells one after another in an arbitrary order. The length of each process is determined by the estimated sensing time, $-\frac{(K-L)\log c}{(M-L)I^*(M,K,L)}$ for each process (i.e., a total length of $-\frac{(K-L)\log c}{I^*(M,K,L)}$). Then, the $M-L$ cells are equally divided to $K-L$ parallel machines, where the total estimated sensing time for all processes in each machine is thus $-\log c/I^*(M,K,L)$. The ordering mechanism is similar to the single target case. A process could be allocated in more than one probing machine, but cannot be observed simultaneously. Go to Step $3$.\vspace{0.2cm}

\noindent
\textbf{3) (Exploitation phase:)} Whenever $N_{H_1}(n)\neq L$ during this step go back to Step $1$. As long as $N_{H_1}(n)= L$, probe cells $\phi(n)$ consecutively following the order in Step $2$, and update the sum LLRs accordingly. For each process, say $j'$, which is ordered in one of probing machines $k=L+1, ..., K$ (which is likely to be normal), take observations consecutively until $S_{j'}(n)$ is smaller than
\begin{equation}
    \frac{(K-L) D(f||g) \log c}{(M-L) I^*(M,K,L)}.
\end{equation}
Then, sample the next process following the order in Step $2$. If the process, say $j$, is ordered in two of probing machines $k \mbox{ and } k+1$, where $k\neq 1,...,L$, then start by taking the first relative portion of the set of observations in machine $k+1$ as divided by the scheduler, referred to as $q_j^{k+1}$, such that $S_j(n)$ is smaller than\vspace{0.1cm}
\begin{equation}
    q_j^{k+1} \cdot \frac{(K-L) D(f||g) \log c}{(M-L) I^*(M,K,L)}.
\end{equation}
When sampling process $j$ in the second probing machine $k$ later, take observations until $S_j(n)$ is smaller than\vspace{0.1cm}
\begin{equation}
    \small{(q_j^{k} + q_j^{k+1})\cdot \frac{(K-L) D(f||g) \log c}{(M-L) I^*(M,K,L)} = \frac{(K-L) D(f||g) \log c}{(M-L) I^*(M,K,L)}}.
\end{equation}
For the processes in probing machines $k=1, ...,L$ (which are likely to be the targets), take observations consecutively until $\Delta_L S(n)$ is greater than $-\log c$. Note that the additional idle times are added only to processes which are suspected to be normal (i.e., allocated in machines $k=L+1,...,K$) to guarantee the feasibility of the sensing process. The completion of this step determines the stopping time of the test $\tau$. The decision rule is given by $\delta = ( m^{(1)}(\tau), m^{(2)}(\tau), ..., m^{(L)}(\tau))$.\vspace{0.1cm}

\subsection{Performance Analysis}
\label{ssec:performance_CCS}

The following theorem establishes the theoretical performance of the CCS algorithm, and particularly its asymptotic optimality in terms of minimizing the Bayes risk as $c \rightarrow 0$ (i.e., as time increases). \vspace{2mm}

\begin{theorem}
\label{th:asymptotic_optimality}
Consider the case of $L\geq 1$ anomalies. Let $R^*$ and $R(\Gamma)$ be the Bayes risks under the CCS policy and any other policy $\Gamma$, respectively. Then, the following statements hold:
\begin{enumerate}
  \item (\emph{Finite sample error bound:}) The detection error probability under CCS is upper bounded by $A\cdot c$ for all $c$, where $A$ is a constant independent of $c$.\vspace{0.2cm}
  \item (\emph{Bounded expected number of switchings:}) The expected number of switchings during the running time of CCS is $O(1)$ with respect to $c$.\vspace{0.2cm}
  \item (\emph{Asymptotic optimality:}) The Bayes risk satisfies\footnote{The notation $f \sim g$ as $c \rightarrow 0$ refers to $\underset{c \rightarrow 0}\lim f/g =1$.}:
  \beq
  R^* \sim \frac{-c\;\log c}{I^*(M,K,L)} \sim \underset{\Gamma}{\inf}\; R(\Gamma) \mbox{  as } c \rightarrow 0.
  \eeq
  \vspace{0.1cm}
\end{enumerate}
\end{theorem}

\begin{proof}
\label{th:asymptotic_optimality_proof}
A detailed proof can be found in the Appendix. We provide here a sketch of the proof. In the analysis, we start by showing that the error probability under CCS is upper bounded by $A\cdot c$ and the total number of switchings is bounded. Then, we show that the Bayes risk $R^*$ under the CCS policy approaches $-c\log c/ I^*(M,K,L)$ as $c$ approaches zero, which is also shown to be the asymptotic lower bound on the Bayes risk. Specifically, the asymptotic optimality of $R^*$ is established based on Lemma \ref{lem:expectedDetectionTime}, showing that the asymptotic expected detection time approaches $-\log c/ I^*(M,K,L)$, while the number of switchings under CCS is bounded following Lemma \ref{lem:tau_m1}. The asymptotic analysis is based on upper bounding the stopping time of CCS (i.e., the detection time) which dominates the order of the Bayes risk. The detection time is determined by the slowest probing machine out of all $K$ machines. The detection time in each machine consists of the total exploration times (shown in Lemma \ref{lem:tau_m1}), exploitation times (shown in Lemma \ref{lem:barTau_Hk}), and the additional delays applied to guarantee that cells are not sampled by two successive machines simultaneously (shown in Lemma \ref{lem:tauj_delay}). Combining the above yields that the expected detection time is upper bounded by $-\log c/I^*(M,K,L)$ asymptotically as $c\rightarrow 0$ (i.e., contributes a cost $-c\log c/I^*(M,K,L)$), as desired.
\end{proof} \vspace{2mm}

We point out that a bounded number of switchings under CCS is of particular significance. It contrasts sharply the linear switching order (which increases with $-\log c$) as commonly seen in active search strategies.

\section{Simulation Results}
\label{sec:simulation}

In this section, we present experimental results to demonstrate the performance of the proposed CCS algorithm in the finite regime. We compared CCS with the classic Chernoff test \cite{chernoff1959sequential}, the DGF policy \cite{cohen2015active}, which was shown to achieve strong performance for tackling the anomaly detection setting considered here with $s=0$ (i.e., when there is no cost for switching), and the Sluggish Procedure A policy (referred to as the Sluggish policy) \cite{vaidhiyan2017neural}, which was designed to tackle dynamic search with multiple plays and switching cost, as considered here. We simulated the following setting with $10^5$ random trials. The observations followed Rayleigh distribution $f \sim \mbox{Rayleigh}(\sigma_f) \mbox{ or } g \sim \mbox{Rayleigh}(\sigma_g)$, depending on whether the single target is absent or present, respectively. The a priori probability of the target being located in cell $m$ was set to $\pi_m = 1/M$ for all $1 \leq m \leq M$. It can be verified that the KL divergence of two Rayleigh distributions $f$ and $g$ is given by: $D(f||g) = 2\log(\frac{\sigma_g}{\sigma_f})+ \frac{\sigma_f^2 - \sigma_g^2}{\sigma_g^2}$. Let $ R(\Gamma)$ be the Bayes risk under policy $\Gamma$ and $R_{LB}= \frac{-c\log c}{I^*(M,K,L)}$ be the asymptotic lower bound on the Bayes risk. To evaluate the performance of the algorithms in terms of the achievable Bayes risk, we define the relative loss with respect to the asymptotic lower bound of the Bayes risk under policy $\Gamma$ by:
\beq
    L(\Gamma) \triangleq \frac{R(\Gamma) - R_{LB}}{R_{LB}}.
    \label{eq:relative loss}
\eeq
\begin{figure}
\begin{subfigure}[t]{.47\linewidth}
    \centering
    {\includegraphics[width=\textwidth]{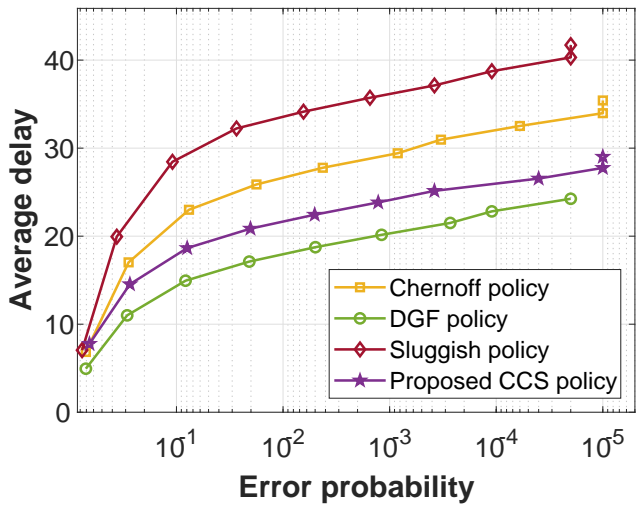}}
    \caption{Average number of observations}
    \label{fig:sim_errK10_a}
\end{subfigure} \hspace{1cm}\vspace{4mm}%
\begin{subfigure}[t]{.47\linewidth}
    \centering
    {\includegraphics[width=\textwidth]{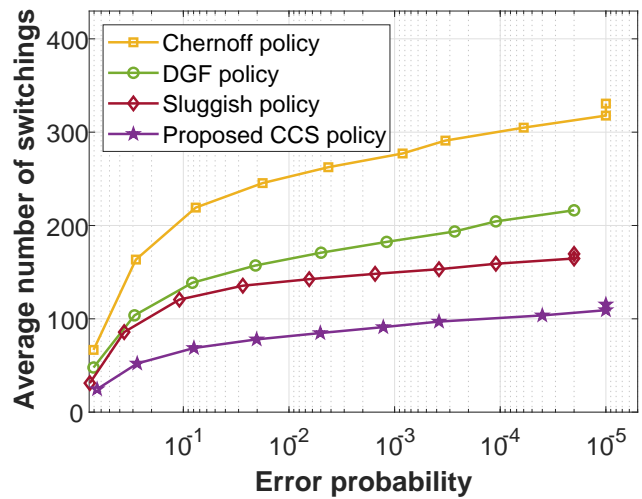}}
    \caption{Average number of switchings}
    \label{fig:sim_errK10_b}
\end{subfigure}\\[1ex]
\begin{subfigure}[t]{.47\linewidth}
    \centering
    {\includegraphics[width=\textwidth]{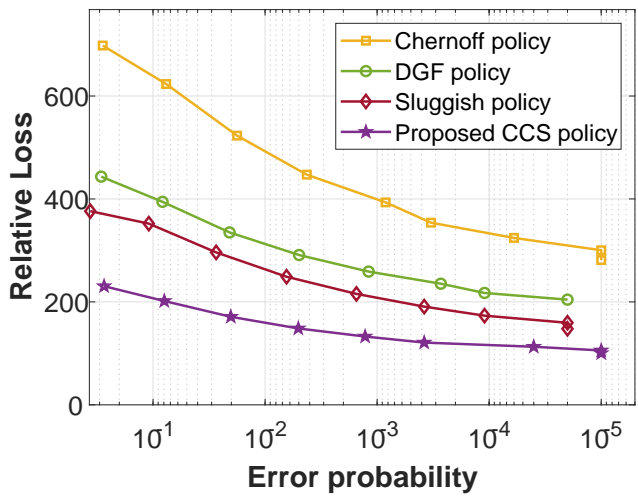}}
    \caption{Relative loss with switching cost, $s=5c$}
    \label{fig:sim_errK10_c}
\end{subfigure}\hspace{1cm}\vspace{4mm}%
\begin{subfigure}[t]{.47\linewidth}
    \centering
    {\includegraphics[width=\textwidth]{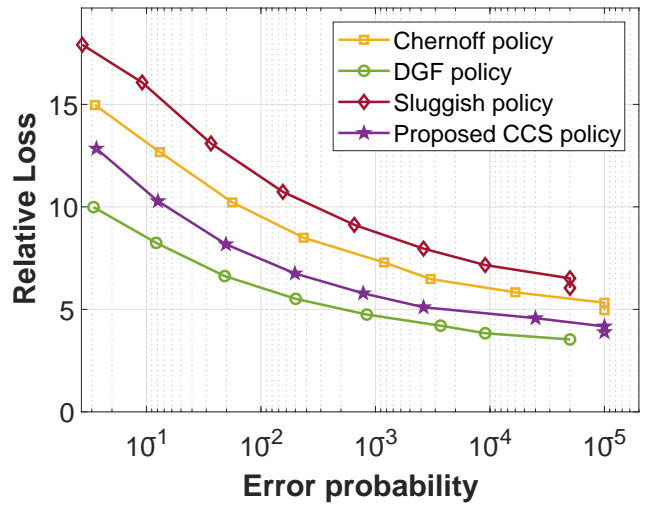}}
    \caption{Relative loss without switching cost, $s=0$}
    \label{fig:sim_errK10_d}
\end{subfigure}
\caption{\small{Performance comparison as a function of the error probability. Simulation parameters: $M=100, K=10, \sigma_f =1, \sigma_g =2$.}}
\end{figure}

We start by simulating a system with the following parameters: $M=100, K=10$, $s=5c$, $\sigma_f=1, \sigma_g = 2$. We obtain $D(g||f) \approx 1.61, D(f||g) \approx 0.64$, so that $D(f||g)/D(g||f)+1 \leq M$ holds. During the exploitation phases, CCS samples the process with the highest sum LLR, $m^{(1)} (n)$, at each given time in probing machine $k=1$, whereas all other processes are sampled consecutively in the rest $(K-1)$ probing machines as dictated by CCS sampling rule (related to the description in Fig. \ref{fig:ordering}). By contrast, the DGF policy samples, at each given time, the $K$ processes with the highest sum LLRs, $m^{(1)}(n), m^{(2)}(n), ..., m^{(K)}(n)$. The Chernoff test applies a randomized strategy that samples the process with the highest sum LLR, $m^{(1)}(n)$, while the other $(K-1)$ processes are sampled randomly (uniformly). The Sluggish policy modifies the Chernoff test, by applying the Chernoff rule or taking the same action as in the previous time step (to avoid switching), depending on a Bernoulli random variable with a tuned parameter used to reduce the total number of switchings. We present the performance of the algorithms in Figs. \ref{fig:sim_errK10_a}, \ref{fig:sim_errK10_b}, \ref{fig:sim_errK10_c}, \ref{fig:sim_errK10_d}. The performance is presented as a function of the error probability (decreasing $c$ decreases the error). It is shown in Fig. \ref{fig:sim_errK10_a} that DGF achieves the best average delay, but CCS almost achieves DGF for small errors. CCS significantly outperforms Chernoff and Sluggish policies for all errors. The significant performance gain of CCS over DGF in terms of average number of switchings can be observed in Fig. \ref{fig:sim_errK10_b}. As expected, the Sluggish policy achieves a lower number of switchings as compared to Chernoff and DGF policies. The relative loss with respect to the Bayes risk (including switching cost), as defined in (\ref{eq:relative loss}), is evaluated in Fig. \ref{fig:sim_errK10_c}. It can be seen that the Chernoff policy performs the worst, and that CCS significantly outperforms all other methods. Finally, we present the relative loss without switching cost as a benchmark for performance in Fig. \ref{fig:sim_errK10_d}. In this setting, the DGF performs the best, while the CCS performs better than the Chernoff policy and Sluggish policy. Nevertheless, CCS almost achieves DGF for small errors even in this setting.
\begin{figure}
\begin{subfigure}{.47\linewidth}
    \centering
    {\includegraphics[width=\textwidth]{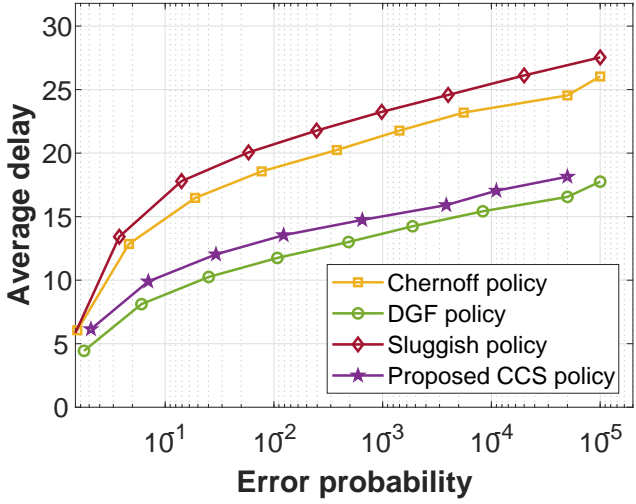}}
    \caption{Average number of observations}
    \label{fig:sim_errK50_a}
\end{subfigure}\hspace{1cm}\vspace{4mm}
\begin{subfigure}{.47\linewidth}
    \centering
    {\includegraphics[width=\textwidth]{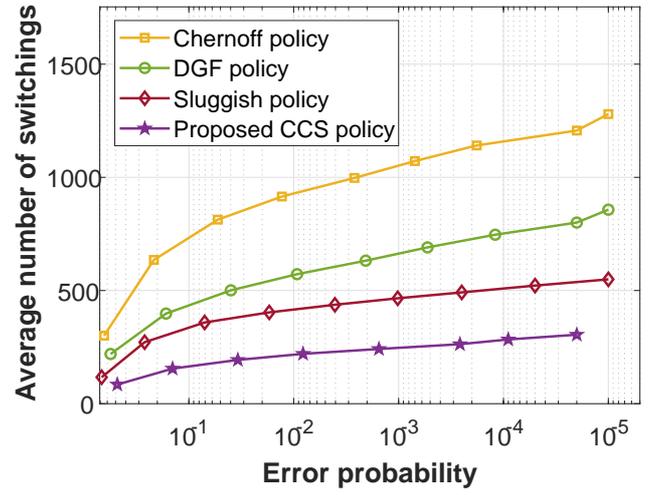}}
    \caption{Average number of switchings}
    \label{fig:sim_errK50_b}
\end{subfigure}\\[1ex]
\begin{subfigure}{.47\linewidth}
    \centering
    {\includegraphics[width=\textwidth]{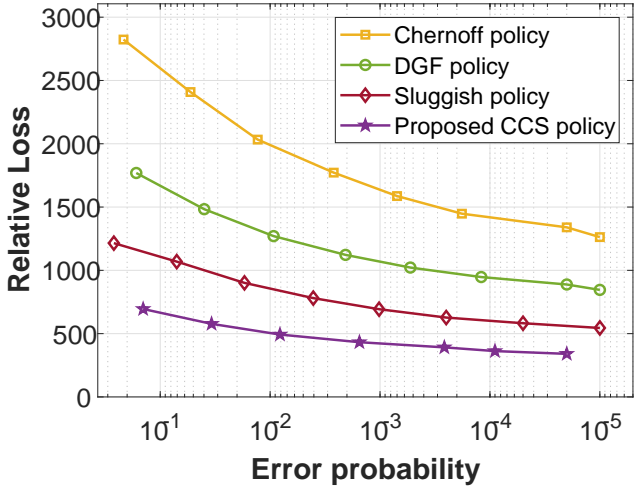}}
    \caption{Relative loss with switching cost, $s=5c$}
    \label{fig:sim_errK50_c}
\end{subfigure}\hspace{1cm}\vspace{4mm}
\begin{subfigure}{.47\linewidth}
    \centering
    {\includegraphics[width=\textwidth]{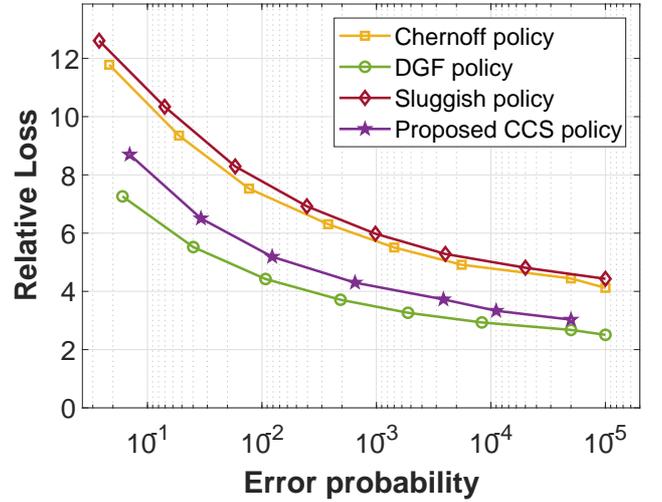}}
    \caption{Relative loss without switching cost, $s=0$}
    \label{fig:sim_errK50_d}
\end{subfigure}
\caption{\small{Performance comparison as a function of the error probability. Simulation parameters: $M=200, K=50, \sigma_f = 4, \sigma_g = 8$.}}
\end{figure}

Next, we increased the system size and set $M=200, K=50$, and set $s = 5c, \sigma_f=4, \sigma_g = 8$. We present the performance of the algorithms in Figs. \ref{fig:sim_errK50_a}, \ref{fig:sim_errK50_b}, \ref{fig:sim_errK50_c}, \ref{fig:sim_errK50_d}. The performance is presented as a function of the error probability (decreasing $c$ decreases the error).
It is shown in Fig. \ref{fig:sim_errK50_a} that DGF and CCS achieve the best average delay (DGF slightly outperforms CCS for large errors), whereas CCS significantly outperforms both Chernoff and Sluggish policies again for all errors. The significant performance gain of CCS over DGF in terms of average number of switchings can be observed again in Fig. \ref{fig:sim_errK50_b}. The relative loss with respect to the Bayes risk, as defined in (\ref{eq:relative loss}), is evaluated in Fig. \ref{fig:sim_errK50_c}. It can be seen that CCS outperforms the other methods including the Sluggish policy. Finally, we present the relative loss without switching cost as a benchmark for performance in Fig. \ref{fig:sim_errK50_d}. As expected, DGF performs the best in this case again. Nevertheless, it only slightly outperforms CCS for small errors. CCS still presents strong performance and outperforms the Sluggish policy and Chernoff policy in this setting. Finally, we ran this setting for a wide value range of the system size $M$ and the monitoring capacity $K$. We present in Figs. \ref{fig:sim_M}, \ref{fig:sim_K} the relative loss as a function of $M$ and $K$, respectively. It can be seen that CCS achieves a significant performance gain over all other methods for all values of $M, K$.

\begin{figure}[h!]
\centering
\begin{subfigure}{.47\columnwidth}
    \centering
    {\includegraphics[width=\linewidth]{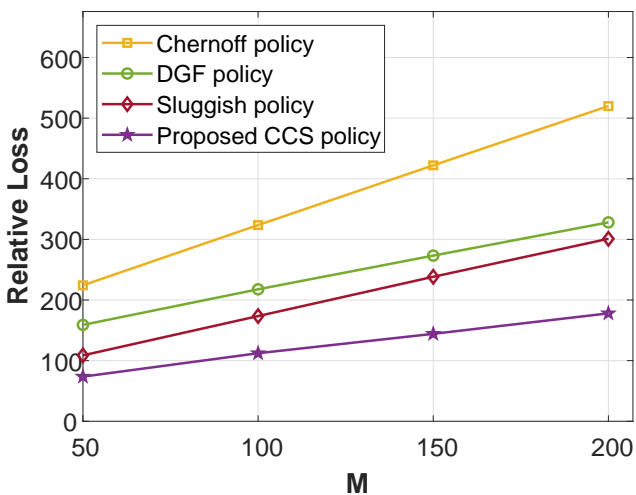}}
    \caption{Performance comparison as a function of the system size $M$, when $K=10$}
    \label{fig:sim_M}
\end{subfigure}
\hfill
\begin{subfigure}{.47\columnwidth}
    \centering
    {\includegraphics[width=\linewidth]{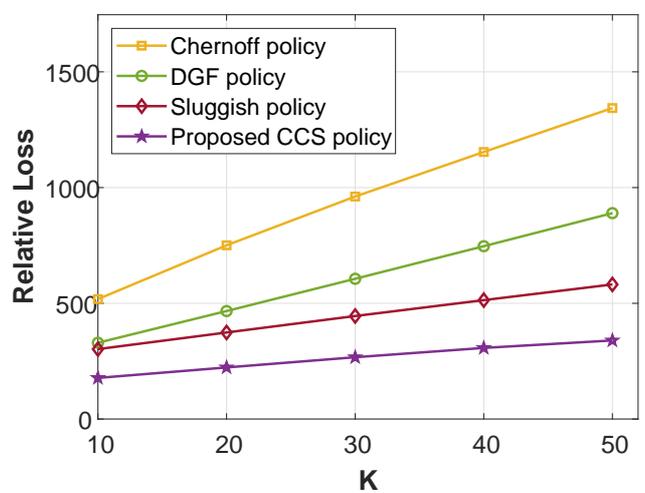}}
    \caption{Performance comparison as a function of the monitoring capacity $K$, when $M=200$}
    \label{fig:sim_K}
\end{subfigure}
\caption{\small{Relative loss with switching cost as a function of the system size and the monitoring capacity.}}
\end{figure}

The results reported above can be explained as follows. While the Chernoff test was shown to be asymptotically optimal in terms of minimizing the Bayes risk without switching cost \cite{chernoff1959sequential}, it suffers from performance degradation in the finite regime as compared to a deterministic strategy as DGF. This can be shown in Figs. \ref{fig:sim_errK10_a}, \ref{fig:sim_errK50_a}, where both DGF and CCS outperforms the Chernoff test. The Sluggish policy modifies the Chernoff test by taking sub-optimal actions from time to time to reduce the switching cost. This decreases performance in terms of the detection delay as can be seen in Figs. \ref{fig:sim_errK10_a}, \ref{fig:sim_errK50_a}. However, it improves the performance in terms of switching cost and total Bayes risk, as can be seen in Figs. \ref{fig:sim_errK10_b}, \ref{fig:sim_errK10_c}, \ref{fig:sim_errK50_b}, \ref{fig:sim_errK50_c}. The DGF policy exploits the structure of the anomaly detection problem to apply a deterministic sampling which minimizes the detection delay. Specifically, it samples, at each given time, the $K$ processes with the highest sum LLRs, $m^{(1)}(n), m^{(2)}(n), ..., m^{(K)}(n)$. This was shown to be asymptotically optimal in terms of minimizing the Bayes risk without switching cost \cite{cohen2015active}. Indeed, it achieves the best detection delay as can be seen in Figs. \ref{fig:sim_errK10_a}, \ref{fig:sim_errK10_d}, \ref{fig:sim_errK50_a}, \ref{fig:sim_errK50_d}. However, since the order of the $K$ processes with the highest sum LLRs changes frequently, DGF results in frequent switchings that significantly degrade the performance in terms of switching cost and consequently the total Bayes risk, as can be seen in Figs. \ref{fig:sim_errK10_b}, \ref{fig:sim_errK10_c}, \ref{fig:sim_errK50_b}, \ref{fig:sim_errK50_c}. The proposed CCS policy exploits the structure of the anomaly detection problem as well, but applies a consecutive sampling to reduce the switching cost. The structure of the consecutive sampling slightly degrades the performance in terms of the detection delay as compared to DGF as can be seen in Figs. \ref{fig:sim_errK10_a}, \ref{fig:sim_errK10_d}, \ref{fig:sim_errK50_a}, \ref{fig:sim_errK50_d}, but significantly improves the performance in terms of switching cost and the total Bayes risk, as can be seen in Figs. \ref{fig:sim_errK10_b}, \ref{fig:sim_errK10_c}, \ref{fig:sim_errK50_b}, \ref{fig:sim_errK50_c}. The results are further validated in Figs. \ref{fig:sim_M}, \ref{fig:sim_K} for various values of $M$ and $K$.
These results demonstrate the importance of using the CCS policy in anomaly detection problems when switching incurs cost.

\section{Conclusions}

The problem of searching for anomalies among multiple processes was studied, where only a subset of the cells can be observed at a time. The objective is to minimize the expected cost due to switching across processes, delay, and detection errors. We proposed the CCS policy to achieve this goal. Its structure allows to sense processes consecutively in exploitation phases and reduce the switching cost, while the exploration time is shown to be bounded. We presented the asymptotic optimality of the proposed policy theoretically, with a bounded accumulated switching cost. In addition, simulation results demonstrated strong performance of CCS in the finite regime as well.

Throughout this paper we assumed that the observation distributions $f, g$, and the number of abnormal processes $L$ are known. A future research direction is to judiciously design and analyze CCS for cases where the observation distributions and the number of abnormal processes are unknown. A possible way to tackle unknown $f, g$ is to use parametric methods to estimate the unknown distribution parameters when the parametric models of $f, g$ are known. When the parametric models of $f, g$ are unknown, non-parametric methods can be incorporated. A possible way to handle an unknown number of abnormal processes $L$ is to design the exploration phases judiciously to infer the value of $L$, where the allocation step in the exploitation phases should be adapted accordingly.

\section{Appendix}

In this appendix we prove the asymptotic optimality of the CCS policy in terms of minimizing the Bayes risk (\ref{eq:bayes_risk}) as $c$ approaches zero, presented in Theorem \ref{th:asymptotic_optimality}. For ease of presentation, we focus on the case of a single anomaly 
($L=1$), and discuss the extension to multiple targets at the end of the proof. Without the loss of generality we prove the theorem when hypothesis $m$ is true. For convenience, we define $\Tilde{\ell}_j(i)$, a zero mean random variable under hypothesis $H_m$ as follows:
\beq
\label{eq:ell_tilda}
\Tilde{\ell}_j(i) \triangleq 
    \begin{cases}
           \ell_j(i)-D(g||f), & \mbox{if } j=m, \\
           \ell_j(i)+D(f||g), & \mbox{if } j\neq m.\\
    \end{cases}
\eeq \vspace{2mm}

\begin{lemma}[Finite sample error bound]
\label{lem:error_upperBound}
The error probability of CCS is upper bounded by
\beq
    \mathbf{P}_e \leq (M-1) c.
    \label{eq:error_upperBound}
\eeq
\end{lemma}\vspace{2mm}

\begin{proof}
The probability of error is defined by
\beq
\mathbf{P}_e = \sum\limits_{m=1}^{M} \pi_m \cdot \alpha_m.
\eeq
Let $\alpha_{m,j}=\mathbf{P}_m(\delta = j)\mbox{ for all } j \neq m$. Thus, the probability of declaring a wrong target cell when hypothesis $H_m$ is true (i.e., $\delta \neq m$) is $\alpha_m =\sum_{j\neq m}\alpha_{m,j}$. 
We define the difference between the observed sum LLR of cells $m$ and $j$ by:
\beq
    \Delta S_{m,j}(n) \triangleq S_m(n)-S_j(n),
\eeq
and define
\beq
    \Delta S_m(n) \triangleq \underset{j\neq m}{\min} \Delta S_{m,j}(n).
\label{eq:deltaSm}
\eeq
Using definitions (\ref{eq:deltaSn}) and (\ref{eq:deltaSm}), we have:
\beq
    \Delta S(n) \triangleq S_{m^{(1)} (n)}(n) - S_{m^{(2)} (n)}(n) = \underset{m}{\max} \Delta S_m(n).
\eeq
Note that accepting hypothesis $H_j$ (i.e., $\Delta S_j(n) \geq -\log c$) implies that $\Delta S_{j,m} \geq -\log c$. By changing the measure, we can show that for all $j \neq m$ we get:
\beq
\bea{l}
    \alpha_{m,j}= \mathbf{P}_m(\delta = j) \vspace{1mm}\\
    = \mathbf{P}_m(\Delta S_j(\tau) \geq -\log c) \leq \mathbf{P}_m(\Delta S_{j,m}(\tau)\geq -\log c) \vspace{1mm}\\
    \leq c \mathbf{P}_j(\Delta S_{j,m}(\tau) \geq -\log c) \leq c.
\ena
\eeq
Finally,
$\alpha_m = \sum\limits_{j\neq m} \alpha_{m,j} \leq (M-1)c$
and equation (\ref{eq:error_upperBound}) follows. \textbf{This proves Statement $1$ of Theorem \ref{th:asymptotic_optimality} for a single anomaly}. 
\end{proof} \vspace{2mm} 
For the following definitions, recall that under hypothesis $H_m$, $S_m(n)$ is a random walk with a positive expected increment \vspace{1mm} $\mathbf{E}_m(\ell_m(n))=D(g||f)>0$ if cell $m$ is the target, while $S_j(n)$ for $j\neq m$ is a random walk with a negative expected \vspace{1mm} increment $\mathbf{E}_m(\ell_j(n))=-D(f||g)<0$. Intuitively, as will be shown analytically in the proof, the sample path of $S_m(n)$ will dominate those of $S_j(n), \; \forall j\neq m$, when $n$ is sufficiently large. For the rest of the proof, we define below two random times, namely $\tau_{m,1} \mbox{ and } \tau_{j,1}$. We start by defining a random time for the true hypothesis $H_m$.\vspace{2mm}

\begin{definition}
Define $\tau_{m,1}$ as the smallest integer such that $S_m(t)>0$ for all $t\geq\tau_{m,1}$. 
\end{definition}\vspace{2mm}

In addition, we define the actual number of measurements $N_{m,1}$ required for obtaining a positive sum LLR for each future sampling times of abnormal process $m$.\vspace{2mm}
\begin{definition}
Denote $N_{m,1}$ as the smallest integer such that $\sum\limits_{i=1}^{r} \ell_m(i)>0$ for all $r\geq N_{m,1}$. 
\end{definition}\vspace{2mm}

For each normal process $j\neq m$, we define the random time $\tau_{j,1}$ similarly for a negative sum LLR, after time $\tau_{m,1}$ has reached.\vspace{2mm}
\begin{definition}
For all $j\neq m$, denote $\tau_{j,1}\geq \tau_{m,1}$ as the smallest integer such that $S_j(t)<0$ for all $t\geq\tau_{j,1}$. 
\end{definition}\vspace{2mm}

In addition, for each normal process $j\neq m$, we define the actual number of measurements $N_{j,1}$ required for obtaining a negative sum LLR for all future sampling times of normal process $j$.\vspace{2mm}
\begin{definition}
Define $N_{j,1}$ as the smallest integer such that $\sum\limits_{i=1}^{r} \ell_j(i)<0$ for all $r\geq N_{j,1}$. 
\end{definition}\vspace{2mm}

\begin{lemma}
\label{lem:tau_m1}
There exists $C>0$ and $\gamma>0$ such that
\beq
    \mathbf{P}_m(\tau_{m,1} >n) \leq C \e^{-\gamma n},
\eeq
for all $m=1,2,...,M$.
\end{lemma}\vspace{2mm}

\begin{proof}
In exploitation phases of Case $1$, a single process whose sum LLR must be positive is probed by machine $k=1$. Therefore, the maximal duration of process $j\neq m$ in machine $1$ is $N_{j, 1}$. In exploitation phases of Case $2$, a process whose sum LLR must be negative is probed by any machine and thus the maximal duration is $N_{m, 1}$. In exploration phases, the processes are probed one by one until only a single sum LLR is positive, and all other sum LLRs are negative. Therefore the entire duration of all exploration phases is upper bounded by $N_{m, 1}+\sum_{j\neq m} N_{j, 1}$. As a result, in Case $1$ $\tau_{m,1}$ is upper bounded by the entire duration in exploitation phases of each normal process being in probing machine 1, plus the entire duration of exploration phases, which yields: $\tau_{m, 1}\leq N_{m, 1}+2\sum_{j\neq m} N_{j, 1}$. For Case $2$, $\tau_{m,1}$ is upper bounded by $2 N_{m, 1}+\sum_{j\neq m} N_{j, 1}$. Next, we show that each term is sufficiently small with high probability. 

Using the definition of $N_{m,1}$ as the total number of samples taken from process $m$ until its sum LLR is strictly positive for all future sampling times, we have:
\beq
\bea{l}
    \mathbf{P}_m(N_{m,1} > n) \leq \mathbf{P}_m(\underset{r\geq n}{\inf} \; \sum\limits_{i=1}^r \ell_m(i) \leq 0) \\
    \leq \sum\limits_{r=n}^{\infty} \mathbf{P}_m(\sum\limits_{i=1}^r -\ell_m(i)\geq 0) \leq
    \sum\limits_{r=n}^{\infty}(E_m[\e^{-s \;\ell_m(1)}])^r ,\vspace{1mm}
\ena
\eeq
for all $s>0$.\\
Note that the moment generating function (MGF) is equal to one at $s=0$ and $\sum\limits_{i=1}^r \ell_m(i)$ is a random walk with a positive expected increment $\mathbf{E}_m(\ell_m(i))=D(g||f)>0$. Differentiating the MGF of $\ell_m(1)$ with respect to $s$ yields a strictly positive derivative at $s=0$. As a result, there exists $s>0$ and $\gamma_1>0$ such that $E_m[\e^{-s \;\ell_m(1)}]< \e^{-\gamma_1} < 1$. 
Hence, by using the Chernoff bound, there exists $C_1>0$ and $\gamma_1>0$ such that for cell $m$:
\beq
\label{eq:Nm1}
\bea{l}
    \mathbf{P}_m(N_{m,1} > n)\leq \sum\limits_{r=n}^{\infty}(E_m[\e^{-s \;\ell_m(1)}])^r \vspace{0.2cm} \\\hspace{24mm}
    \leq \sum\limits_{r=n}^{\infty} \e^{-\gamma_1 r} \leq C_1 \e^{-\gamma_1 n}.
\ena
\eeq

Similarly, using the definition of $N_{j,1}$ as the total number of samples taken from process $j$ until its sum LLR is strictly negative for all future sampling times, then for a normal process $j\neq m$:
\beq
\bea{l}
    \mathbf{P}_m(N_{j,1} > n) \leq \mathbf{P}_m(\underset{r\geq n}{\sup}\; \sum\limits_{i=1}^r \ell_j(i) \geq 0) \\
    \leq \sum\limits_{r=n}^{\infty} \mathbf{P}_m(\sum\limits_{i=1}^r \ell_j(i)\geq 0) \leq \sum\limits_{r=n}^{\infty}(E_m[\e^{s \;\ell_j(1)}])^r, \vspace{1mm}
\ena
\eeq
for all $s>0$.\\
For $j\neq m$, $\sum\limits_{i=1}^r \ell_j(i)$ is a random walk with a negative expected increment $\mathbf{E}_m(\ell_j(i))=-D(f||g)<0$. Differentiating the MGF of $\ell_j(1)$ with respect to $s$ yields strictly negative derivative at $s=0$. As a result, there exists $s>0$ and $\gamma_2>0$ such that $E_m[e^{s \; \ell_j(1)}]< \e^{-\gamma_2} < 1$. Hence, by using the Chernoff bound, there exists $C_2>0$ and $\gamma_2>0$ such that:
\beq
\label{eq:Nj1}
\bea{l}
    \mathbf{P}_m(N_{j,1} > n) \leq \sum\limits_{t=n}^{\infty}(E_m[\e^{s \; \ell_j(1)}])^r \vspace{0.2cm} \\\hspace{22mm}
    \leq \sum\limits_{r=n}^{\infty} \e^{-\gamma_2 r} \leq C_2 \e^{-\gamma_2 n}.
\ena
\eeq

For Case $1$ we finally have:
\beq
\label{eq:tau_m1}
\bea{l}
    \mathbf{P}_m(\tau_{m,1} >n) \leq \mathbf{P}_m(N_{m,1} + 2\sum\limits_{j\neq m} N_{j,1} > n) \vspace{1mm}\\
    \leq \mathbf{P}_m(N_{m,1} > \frac{n}{2M-1}) + 2\sum\limits_{j\neq m} \mathbf{P}_m(N_{j,1} > \frac{n}{2M-1})\vspace{1mm}\\
    \leq C_1 \e^{-\gamma_1 \frac{n}{2M-1}} + 2(M-1)\cdot C_2 \e^{-\gamma_2 \frac{n}{2M-1}} \vspace{1mm}\\
    \leq C \e^{-\gamma n},
\ena
\eeq
where $C\triangleq C_1+2(M-1)C_2>0 \mbox{ and } \gamma\triangleq \min(\frac{\gamma_1}{2M-1}, \frac{\gamma_2}{2M-1})>0$. 
A similar exponential decay can be easily obtained for Case 2 with the corresponding constants, which completes the proof.
\end{proof}\vspace{1mm}

We can deduce from the lemma above that for all $t>\tau_{m,1}$ in Case $1$ the target is always scheduled on machine $k=1$, where the probing order of all normal processes on machines $k=2, ..., K$ remains fixed. It is still possible that the sum LLR of a normal process will be positive at some time $t\geq\tau_{m,1}$, and then an exploration phase will be executed again. Nevertheless, since the sum LLR of the target is positive for all $t\geq\tau_{m,1}$, the target will be scheduled on machine $k=1$ when entering an exploitation phase again, and the probing order on machines $k\neq 1$ will remain the same. As a result, during time steps 
$t\geq\tau_{m,1}$, the number of switchings must be less than $\sum_{j\neq m} N_{j, 1}$. This is because at each time an exploration phase begins, samples are taken one by one from all processes. Similarly, in Case $2$, the number of switchings during exploration must also be less than $\sum_{j\neq m} N_{j, 1}$. Hence, the total number of switchings is bounded $\tau_{m,1}+\sum_{j\neq m} N_{j, 1}$. Therefore, applying Lemma \ref{lem:tau_m1} and (\ref{eq:Nj1}) yields a bounded expected number of switchings. \textbf{This proves Statement $2$ of Theorem \ref{th:asymptotic_optimality} for a single anomaly}.

We next define an additional random time for normal processes in Case $1$. The expansion to Case $2$ of the following definitions and Lemma \ref{lem:Nj} can be done with minor modifications when the sum LLRs for normal processes fulfill $S_j(t) \leq \frac{K \cdot D(f||g) \log c}{(M-1)I^*(M,K,L)}, \forall j\neq m$.\vspace{0.2cm}

\begin{definition}
For all $j\neq m$, denote $\tau_j\geq \tau_{m,1}$ as the smallest integer such that $S_j(t) \leq \frac{(K-1)D(f||g) \log c}{(M-1)I^*(M,K,L)}$ for all $t\geq \tau_j\geq \tau_{m,1}$.
\end{definition}
\vspace{2mm}

Note that $\tau_j$ is the last time process $j\neq m$ is probed during exploitation phases. \vspace{2mm}

In addition, for each normal process $j\neq m$, we define the actual number of measurements $N_j$ required for obtaining a sum LLR smaller than $\frac{(K-1)D(f||g) \log c}{(M-1)I^*(M,K,L)}$\vspace{1mm} for all future sampling times of normal process $j$ in Case $1$.\vspace{2mm}

\begin{definition}
\label{def:Nj}
Define $N_j$ as the smallest integer such that $\sum\limits_{i=1}^{r} \ell_j(i)\leq \frac{(K-1)D(f||g) \log c}{(M-1)I^*(M,K,L)}$ for all $r\geq N_j$. 
\end{definition}\vspace{2mm}

We next analyze major events for gathering enough information to distinguish between the target cell $m$ and all normal cells $j\neq m$. Lemma \ref{lem:Nj} below shows that $N_j$ in Case $1$ is less than $\sim-\frac{(K-1)\log c}{(M-1)I^*(M,K,L)}$ with high probability. We will use this result later to upper bound the total detection time.
\vspace{2mm}

\begin{lemma}
\label{lem:Nj}
Consider a normal process $j$ probed indefinitely. Then, for every fixed $\epsilon>0$ there exist $C>0$ and $\gamma>0$ such that in Case $1$:
\beq
    \mathbf{P}_m(N_j >n) \leq C \e^{-\gamma n}\vspace{1mm} \;\forall n> -(1+\epsilon)\frac{(K-1)\log c}{(M-1)I^*(M,K,L)},
\eeq
for all $m=1,2,...,M$.
\end{lemma} \vspace{2mm}

\begin{proof}
By the definition of $N_j$, we have
\beq
\bea{l}
    \mathbf{P}_m(N_j >n) \vspace{1mm}\\
    \hspace{12mm}\leq \mathbf{P}_m \bigg( \underset{r \geq n}{\sup} \; \sum\limits_{v=1}^{r} \ell_{j}(v) \geq \frac{(K-1)D(f||g) \log c}{(M-1)I^*(M,K,L)} \bigg).
\ena
\eeq
For any $\epsilon> 0$ there exists $\epsilon_1 = D(f||g)\cdot\epsilon/(1+\epsilon) > 0$ such that:
\beq
\bea{l}
    \sum\limits_{v=1}^{r} \ell_{j}(v) - \frac{(K-1)D(f||g) \log c}{(M-1)I^*(M,K,L)} \vspace{2mm}\\ 
    \hspace{8mm}= \sum\limits_{v=1}^{r} \Tilde{\ell}_{j}(v) - r D(f||g) - \frac{(K-1)D(f||g) \log c}{(M-1)I^*(M,K,L)} \vspace{2mm}\\ 
    \hspace{8mm}= \sum\limits_{v=1}^{r} \Tilde{\ell}_{j}(v) - r D(f||g) \bigg[1- \frac{-(K-1)\log c}{r(M-1)I^*(M,K,L)}\bigg] \vspace{2mm}\\ 
    \hspace{8mm}\leq \sum\limits_{v=1}^{r} \Tilde{\ell}_{j}(v) - r \epsilon_1,
\ena
\eeq
for all $r\geq n>-(1+\epsilon)\frac{(K-1)\log c}{(M-1)I^*(M,K,L)}$, where $\Tilde{\ell}_{j}(v)$ is defined in (\ref{eq:ell_tilda}).\\
As a result,
\beq
    \sum\limits_{v=1}^{r} \ell_{j}(v) \geq \frac{(K-1)D(f||g) \log c}{(M-1)I^*(M,K,L)}
\eeq
implies 
\beq
    \sum\limits_{v=1}^{r} \Tilde{\ell}_{j}(v) \geq r\epsilon_1.
\eeq
Hence, for any $\epsilon>0 $ there exists $\epsilon_1 > 0$ such that:
\beq
\bea{l}
    \mathbf{P}_m(N_j >n) \leq \mathbf{P}_m \bigg( \underset{r \geq n}{\sup} \; \sum\limits_{v=1}^{r} \ell_{j}(v) \geq \frac{(K-1)D(f||g) \log c}{(M-1)I^*(M,K,L)} \bigg)\vspace{1mm}\\
    \hspace{21mm}\leq \sum\limits_{r=n}^\infty \mathbf{P}_m \bigg( \sum\limits_{v=1}^{r} \ell_{j}(v) \geq \frac{(K-1)D(f||g) \log c}{(M-1)I^*(M,K,L)}\bigg) \vspace{1mm}\\ 
    \hspace{21mm}\leq \sum\limits_{r=n}^\infty \mathbf{P}_m \bigg( \sum\limits_{v=1}^{r} \Tilde{\ell}_{j}(v) \geq r\epsilon_1 \bigg),
\ena
\eeq
for all $r\geq n>-(1+\epsilon)\frac{(K-1)\log c}{(M-1)I^*(M,K,L)}$.\vspace{1mm}\\
By applying the Chernoff bound, it can be shown that there exists $\gamma_1 >0$ such that: \vspace{1mm} $\mathbf{P}_m(\sum\limits_{v=1}^{r} \Tilde{\ell}_{j}(v) \geq r\epsilon_1) \leq \e^{-\gamma_1 r} \mbox{ for all } r\geq n>-(1+\epsilon)\frac{(K-1)\log c}{(M-1)I^*(M,K,L)}$.\vspace{1mm} Thus, there exist $C_1>0 \mbox{ and } \gamma_1>0 \mbox{ such that } \vspace{1mm} \mathbf{P}_m(N_j >n) \leq C_1 \e^{-\gamma_1 n} \mbox{ for all } n > -(1+\epsilon)\frac{(K-1)\log c}{(M-1)I^*(M,K,L)}$\vspace{1mm} as desired.\\
\end{proof} \vspace{2mm}

To modify Lemma \ref{lem:Nj} to Case $2$, we define $N_j$ as the number of measurements required for obtaining a sum LLR smaller than $\frac{K D(f||g) \log c}{(M-1)I^*(M,K,L)}$ for all future sampling times of normal process $j$. Then, using similar steps as specified above yields that for every fixed $\epsilon>0$ there exist $C>0$ and $\gamma>0$ such that:
\beq
\label{eq:lemma3_case2}
    \mathbf{P}_m(N_j >n) \leq C \e^{-\gamma n}\vspace{1mm} \;\forall n> -(1+\epsilon)\frac{K\log c}{(M-1)I^*(M,K,L)},
\eeq
for all $m=1,2,...,M$.\vspace{2mm}

In the sequel, we use $N_j$ to upper bound the total time spent during all exploitation phases of normal cell $j$.

\vspace{2mm}

\begin{definition}
Denote $\mathcal{H}^{k}\subseteq\{1,...,M\}$ as the set of indices, indicating which cells are ordered for probing by machine $k$ at time $t>\tau_{m,1}$ in exploitation phases.
\end{definition}\vspace{2mm}

Note that $\mathcal{H}^{k}$ remains fixed in exploitation phases for all $t\geq\tau_{m,1}$ for all $k$.\vspace{2mm}

\begin{definition}
Define the random variable $\bar{N}^{k}$ by: $\bar{N}^{k} \triangleq \sum\limits_{j\in \mathcal{H}^{k}} q_j^{k}\cdot N_j$.
\end{definition} \vspace{2mm}

Note that the r.v. $\bar{N}^{k}$ upper bounds the total time spent by sampling normal processes during all exploitation phases. This will be used later to upper bound the stopping time $\tau$ of CCS.\vspace{2mm}

\begin{lemma}
\label{lem:barTau_Hk}
Assume that CCS is implemented for all $t> \tau_{m,1}$. Then, for every fixed $\epsilon>0$ there exist $C>0$ and $\gamma>0$ such that for all machines $k\neq 1$ we have:
\beq
    \mathbf{P}_m(\bar{N}^{k} >n) \leq C\e^{-\gamma n} \; \forall n> -(1+\epsilon)\log c/I^*(M,K,L), \vspace{1mm}
\eeq
for all $m=1,2,...,M$.
\end{lemma} \vspace{2mm}

\begin{proof}
\label{lem:barTau_Hk_proof}
In Case $1$, after time $t>\tau_{m,1}$, $\;M-1$ normal processes are ordered in $K-1$ probing machines, such that in each probing machine cells are probed according to the set $\mathcal{H}^{k}$ for all $k\neq 1$. According to the ordering mechanism, for each probing machine $k=2, ...,K$ it holds:
\beq
\label{eq:sum_qj}
    \sum\limits_{j\in\mathcal{H}^{k}} q_j^{k} = \frac{M-1}{K-1}.
\eeq
Note that $q_j^k\cdot N_j$ measurements are taken from cell $j$ in machine $k$, where the upper bound of all measurements taken by machine $k$ is given by $\bar{N}^k$. By the definition of $\bar{N}^{k}$ we have:
\beq
\label{eq:barNk_dev}
\bea{l}
    \mathbf{P}_m(\bar{N}^{k} >n) 
    = \mathbf{P}_m(\sum\limits_{j\in \mathcal{H}^{k}} q_j^{k} \cdot N_j >n) \\\vspace{1.5mm}
    \hspace{21mm}= \mathbf{P}_m(\sum\limits_{j\in \mathcal{H}^{k}} q_j^{k} \cdot N_j > \sum\limits_{j\in \mathcal{H}^{k}} q_j^{k}\frac{K-1}{M-1} n) \\\vspace{1.5mm}
    \hspace{21mm}\leq \sum\limits_{j\in \mathcal{H}^{k}} \mathbf{P}_m(q_j^{k} N_j >q_j^{k}\frac{K-1}{M-1} n),
\ena
\eeq
where the second equality follows by applying (\ref{eq:sum_qj}).
By using Lemma \ref{lem:Nj}, there exist $C_j>0 \mbox{ and } \gamma_j$ such that:
\beq
\label{eq:qjk_Nj_samples}
    \mathbf{P}_m( q_j^k N_j >n)\leq C_j \e^{-\gamma_j n},
\eeq
for all $n>-(1+\epsilon)q_j^{k} \frac{(K-1)\log c}{(M-1)I^*(M,K,L)}$.
Thus, 
\beq
\label{eq:qjk_nj}
    \mathbf{P}_m(q_j^{k} N_j>q_j^{k} \frac{K-1}{M-1} n) \leq C_j \e^{-\gamma_j q_j^k\frac{K-1}{M-1} n},
\eeq
for all $n>-(1+\epsilon)\log c/I^*(M,K,L)$.
Finally, by combining (\ref{eq:barNk_dev}) and (\ref{eq:qjk_nj}), there exist $C\triangleq \sum\limits_{j\in \mathcal{H}^{k}} C_j$ and $\gamma\triangleq \underset{j}{\min}\{\gamma_j\cdot q_j^k\frac{K-1}{M-1}\}$ such that:
\beq
    \mathbf{P}_m(\bar{N}^{k} >n) \leq C \e^{-\gamma n},
\eeq
for all $n>-(1+\epsilon)\log c/I^*(M,K,L)$.

Similarly, in Case $2$ after time $t>\tau_{m,1}$, $\;M-1$ normal processes are ordered in $K$ probing machines. According to the probing mechanism, for each probing machine $k$, it holds that: 
\beq
\label{eq:sum_qj_2}
    \sum\limits_{j\in\mathcal{H}^{k}} q_j^{k} = \frac{M-1}{K},
\eeq
which leads:
\beq
\bea{l}
    \mathbf{P}_m(\bar{N}^{k} >n) = \mathbf{P}_m(\sum\limits_{j\in \mathcal{H}^{k}} q_j^{k} \cdot N_j > \sum\limits_{j\in \mathcal{H}^{k}} q_j^{k}\frac{K}{M-1} n) \\\vspace{1.5mm}
    \hspace{21mm}\leq \sum\limits_{j\in \mathcal{H}^{k}} \mathbf{P}_m(q_j^{k} N_j >q_j^{k}\frac{K}{M-1} n).
\ena
\eeq
Using the modification of Lemma \ref{lem:Nj} to Case $2$ presented in \eqref{eq:lemma3_case2}, there exist $C_j>0 \mbox{ and } \gamma_j$ such that:
\beq
    \mathbf{P}_m( q_j^k N_j >n)\leq C_j \e^{-\gamma_j n},
\eeq
for all $n>-(1+\epsilon)q_j^{k} \frac{K\log c}{(M-1)I^*(M,K,L)}$. The rest of the proof follows directly.
\end{proof} 
\vspace{2mm}

In what follows we define a random time $\tau_m>\tau_{m,1}$. $\tau_m$ can be viewed as the time when sufficient information has been gathered from cell $m$, to later distinguish between hypothesis $m$ and all false hypotheses $j\neq m$ in Case $1$. We point out that $\tau_m$ is not a stopping time. Note that the stopping time $\tau$ depends also on the information gathered from all normal cells, and more specifically information gathered from cell $m^{(2)}(\tau)$.\vspace{2mm}

\begin{definition}
For the target cell $m$, denote $\tau_m$ as the smallest integer such that $S_m(t)\geq -\frac{D(g||f)\log c}{I^*(M,K,L)}$ for all $t\geq\tau_m$.
\end{definition} \vspace{2mm} 

We next define the number of measurements $N_m$ to gather sufficient information from target cell $m$.\vspace{2mm} 

\begin{definition}
For the target cell $m$, denote $N_m$ as the smallest integer such that $\sum\limits_{i=1}^{r} \ell_m(i) \geq -\frac{D(g||f)\log c}{I^*(M,K,L)}$ for all $r\geq\ N_m$.
\end{definition} \vspace{2mm} 

In the sequel we use $N_m$ to bound the detection time spent over machine $k=1$ in Case $1$.\vspace{2mm} 

\begin{lemma}
\label{lem:N_m}
Assume that the target process $m$ is probed indefinitely by machine $k=1$ for all $t>\tau_{m,1}$ in Case $1$. Then, for every fixed $\epsilon>0$ there exist $C>0$ and $\gamma>0$ such that
\beq
    \mathbf{P}_m(N_m >n) \leq C \e^{-\gamma n}\;\; \forall n>-(1+\epsilon)\log c/I^*(M,K,L),
\eeq
for all $m=1,2,...,M$. 
\end{lemma} 
\vspace{2mm}

\begin{proof}
\label{lem:tau_m_proof}
By the definition of $N_m$ as the smallest integer such that $\sum\limits_{i=1}^{r} \ell_m(i) \geq -\frac{D(g||f)\log c}{I^*(M,K,L)}$ \vspace{1.3mm} for all $r\geq N_m$, we get
\beq
    \mathbf{P}_m(N_m >n) \leq \mathbf{P}_m \bigg(\underset{r \geq n}{\inf} \; \sum\limits_{v=1}^{r} \ell_{m}(v) \leq -\frac{D(g||f) \log c}{I^*(M,K,L)} \bigg).
\eeq
For any $\epsilon > 0$ there exists $\epsilon_1=D(g||f)\cdot\epsilon/(1+\epsilon) > 0$ such that:
\beq
\bea{l}
    \sum\limits_{v=1}^{r} \ell_{m}(v) + \frac{D(g||f) \log c}{I^*(M,K,L)} \vspace{2mm}\\
    \hspace{8mm}= \sum\limits_{v=1}^{r} \Tilde{\ell}_{m}(v) + r D(g||f) + \frac{D(g||f) \log c}{I^*(M,K,L)} \vspace{2mm}\\
    \hspace{8mm}\geq \sum\limits_{v=1}^{r} \Tilde{\ell}_{m}(v) + r \epsilon_1 ,
\ena
\eeq
for all $r\geq n>-(1+\epsilon)\log c /I^*(M,K,L)$, where $\Tilde{\ell}_{m}(v)$ is defined in (\ref{eq:ell_tilda}).\\
As a result,
\beq
    \sum\limits_{v=1}^{r} \ell_{m}(v) \leq -\frac{D(g||f)\log c }{I^*(M,K,L)}
\eeq
implies 
\beq
    \sum\limits_{v=1}^{r} \Tilde{\ell}_{m}(v) \leq -r \epsilon_1.
\eeq
Hence, for any $\epsilon>0 $ there exists $\epsilon_1 > 0 $ such that:
\beq
\bea{l}
    \mathbf{P}_m(N_m >n) \leq \mathbf{P}_m \bigg( \underset{r \geq n}{\inf} \; \sum\limits_{v=1}^{r} \ell_{m}(v) \leq -\frac{D(g||f)\log c}{I^*(M,K,L)} \bigg)\vspace{1mm}\\ 
    \hspace{22mm}\leq \sum\limits_{r=n}^\infty \mathbf{P}_m \bigg( \sum\limits_{v=1}^{r} \ell_{m}(v) \leq -\frac{D(g||f) \log c}{I^*(M,K,L)}\bigg) \vspace{1mm}\\ 
    \hspace{22mm} \leq \sum\limits_{r=n}^\infty \mathbf{P}_m \bigg( \sum\limits_{v=1}^{r} -\Tilde{\ell}_{m}(v) \geq r \epsilon_1 \bigg),
\ena
\eeq
for all $r\geq n> -(1+\epsilon)\log c/I^*(M,K,L)$. \\
By applying the Chernoff bound, it can be shown that there exists $\gamma_1 >0$ such that: \vspace{1.5mm} $\mathbf{P}_m\bigg ( \sum\limits_{v=1}^{r} -\Tilde{\ell}_{m}(v) \geq r \epsilon_1 \bigg) \leq \e^{-\gamma_1 r} \mbox{ for all } r \geq n > -(1+\epsilon)\log c/I^*(M,K,L)$.\vspace{1.5mm} Thus, there exist $C_1>0 \mbox{ and } \gamma_1>0 \mbox{ such that } \mathbf{P}_m(N_m >n) \leq C_1 \e^{-\gamma_1 n} \mbox{ for all } n > -(1+\epsilon)\log c/I^*(M,K,L)$. \vspace{1.5mm} \\
\end{proof}

We next define a random time for normal processes, probed by machine $k\neq 1$ in Case $1$ and by machine $k$ in Case $2$ for time $t>\tau_{m,1}$. Recall that a cell cannot be probed by more than one probing machine at a time. Therefore, in a case of a conflict where two successive machines are scheduled to probe the same process simultaneously, an additional idle time would be added to prevent this situation and guarantee the feasibility of the sensing process. An illustration can be found in Fig. \ref{fig:idled_delay}, presented in Section \ref{sec:CCSpolicy}. \vspace{2mm}

\begin{definition}
Consider normal processes $j',j$, such that for all $t>\tau_{m,1}$, the probing order on machine $k\neq 1$ in Case $1$ and on machine $k$ in Case $2$ is process $j'$ followed by a split process $j$ (process $j$ is ordered to be sampled by both machines $k$ and $k+1$). For all $j\neq m$, denote $D_{j'}^{k}$ as the additional idled time added to process $j'$ in probing machine $k$ to prevent cell $j$ being probed simultaneously by machines $k$ and $k+1$.
\end{definition}\vspace{2mm}

In the lemma below we show that $D_{j'}^{k}$ is sufficiently small with high probability. Since this statement holds for all processes $j'\neq m$ which are ordered to be probed by two successive machines, this results in a bounded expected idled time spent in the entire running time.\vspace{2mm}

\begin{lemma}
\label{lem:tauj_delay}
Consider normal processes $j',j$, such that for all $t>\tau_{m,1}$, the probing order on machine $k\neq 1$ in Case $1$ and on machine $k$ in Case $2$ is process $j'$ followed by a split process $j$ (process $j$ is ordered to be sampled by both machines $k$ and $k+1$). Then, there exist $C>0$ and $\gamma>0$ such that for all $k\neq 1$ in Case $1$, $k$ in Case $2$ and $j'\neq m$
\beq
    \mathbf{P}_m(D_{j'}^{k} > n) \leq C \e^{-\gamma n},
\eeq
for all $m=1,2,...,M$.
\end{lemma} \vspace{2mm}

\begin{proof}
For convenience, throughout the proof we refer to the process ordering as in Fig. \ref{fig:idled_delay}. In general, an idle time is added to process $j'$ when the time process $j$ in machine $k+1$ reaches its desired sum LLR is longer than the time process $j'$ reaches its desired sum LLR. By the construction of the ordering scheduler, when process $j'$ is followed by a split process $j$ in machine $k$, $q_{j'}^{k}\geq q_j^{k+1}$ is fulfilled (an illustration to demonstrate the relation between $q_{j'}^{k} \mbox{ and } q_j^{k+1}$ can be found in Fig. \ref{fig:ordering} at ordering time $t=16$ for machines $k=2,3,4$). The case where $q_{j'}^{k}=1$ is specific, and therefore we will refer to a general case $q_{j'}^{k}\leq 1$ for the remaining of the proof. 

We define the total time between the last time process $j$'s sum LLR is smaller than $q_j^{k+1}\cdot\frac{(K-1)D(f||g)\log c}{(M-1)I^*(M,K,L)}$ and the first time process $j'$'s sum LLR is smaller than $q_{j'}^{k}\cdot\frac{(K-1)D(f||g)\log c}{(M-1)I^*(M,K,L)}$. Let $\tilde{N}_{j'}$ be the number of measurements required for obtaining a sum LLR smaller than $\frac{(K-1)D(f||g)\log c}{(M-1)I^*(M,K,L)}$ for the first time, when sampling normal process $j'$. By definition \ref{def:Nj}, $N_j$ is the number of measurements required for obtaining a sum LLR smaller than $\frac{(K-1)D(f||g)\log c}{(M-1)I^*(M,K,L)}$ for all future sampling times of normal process $j$.
Therefore, we use the two r.v. $q_j^{k+1}\cdot N_j$ and $q_{j'}^k \cdot \tilde{N}_{j'}$ as the number of measurements taken from processes $j \mbox{ and } j'$ to reach their desired sum LLRs, when both processes are sampled indefinitely for $t>\tau_{m,1}$. The entire idle time between the two processes is bounded by the number of measurements $q_j^{k+1}\cdot N_j-q_{j'}^k \tilde{N}_{j'}$.

Next, we define $n_j\triangleq -(1+\epsilon)q_j^{k+1}\frac{(K-1)\log c}{(M-1)I^*(M,K,L)}$, and $\tilde{n}_{j'}\triangleq -(1+\epsilon)q_{j'}^{k}\frac{(K-1)\log c}{(M-1)I^*(M,K,L)}$. 
Then,
\beq
\label{eq:Dj'k}
\bea{l}
    \mathbf{P}_m(D_{j'}^{k} > n) \leq \mathbf{P}_m(q_j^{k+1}\cdot N_j-q_{j'}^k \cdot \tilde{N}_{j'} > n) \vspace{1mm}\\
    \leq \mathbf{P}_m(q_j^{k+1}\cdot N_j>n_j + \frac{n}{2}) +\mathbf{P}_m (q_{j'}^k \cdot\tilde{N}_{j'}< n_j - \frac{n}{2}) \vspace{1mm}\\
    \leq \mathbf{P}_m(q_j^{k+1}\cdot N_j>n_j + \frac{n}{2}) +\mathbf{P}_m (q_{j'}^k \cdot \tilde{N}_{j'}<\tilde{n}_{j'} - \frac{n}{2}),\vspace{1mm}
\ena
\eeq 
where the last inequality follows since $\tilde{n}_{j'}\geq n_j$, as $q_{j'}^{k}\geq q_j^{k+1}$. Note that the r.vs $q_j^{k+1}\cdot N_j$, and $q_{j'}^k \cdot\tilde{N}_{j'}$ are concentrated around the values $n_j$, and $\tilde{n}_{j'}$, respectively. Therefore, developing both terms as in Lemma \ref{lem:Nj} and (\ref{eq:qjk_Nj_samples}), we obtain that both terms decrease exponentially with $n$ which proves the lemma for Case 1.
The proof can be easily shown for Case $2$ with minor modifications, by defining $\tilde{N}_{j'}$ and $N_j$ as the number of measurements required for obtaining a sum LLR smaller than $\frac{K D(f||g)\log c}{(M-1)I^*(M,K,L)}$ for the first time and for all future sampling times of normal process $j$, respectively. Then, applying Lemma \ref{lem:Nj}, and defining $n_j$ and $\tilde{n}_{j'}$ such that $n_j\triangleq -(1+\epsilon)q_j^{k+1}\frac{K\log c}{(M-1)I^*(M,K,L)}$, and $\tilde{n}_{j'}\triangleq -(1+\epsilon)q_{j'}^{k}\frac{K\log c}{(M-1)I^*(M,K,L)}$, completes the proof.
\end{proof}\vspace{2mm}

Finally, we conclude that the detection time of CCS is asymptotically bounded by $-\log c/I^*(M,K,L)$, as shown in the next lemma. This will lead to the expected detection time $\tau$, desired for proving the asymptotic optimality of CCS. \vspace{2mm}

\begin{lemma}
\label{lem:expectedDetectionTime}
The expected detection time $\tau$ under the CCS policy satisfies
\beq
    \mathbf{E}_m(\tau) \leq -(1+o(1)) \frac{\log c}{I^*(M,K,L)},
\eeq
for $m=1,2,...,M$.
\end{lemma} \vspace{2mm}

\begin{proof}
\label{lem:expectedDetectionTime_proof}
Define $\tau^{k}$ as the total time spent on machine $k$. Thus, the stopping time $\tau$ of CCS is determined by the slowest machine, and is given by:
\beq 
    \tau \leq \max\{\tau^{1},\tau^{2},...,\tau^{K}\}.
\eeq
It remains to show that the probability $\mathbf{P}_m(\max\{\tau^{1},\tau^{2},...,\tau^{K}\} >n)$\vspace{1mm} decreases exponentially for all $n>-(1+\epsilon)\log c/I^*(M,K,L)$. Note that:
\beq
\label{eq:max_totalDelay_k}
    \mathbf{P}_m(\max\{\tau^{1},\tau^{2},...,\tau^{K}\} >n) \leq \sum\limits_{k=1}^{K} \mathbf{P}_m(\tau^{k} >n).
\eeq

For machine $k=1$ in Case $1$, the total time spent on the machine is bounded by $\tau^1\leq \tau_{m,1}+N_m$. Hence, by combining Lemmas \ref{lem:tau_m1} and \ref{lem:N_m}, for every fixed $\epsilon_1>0$ we get: 
\beq
\label{eq:tau_k_1}
\bea{l}
    \mathbf{P}_m(\tau^{1} >n)\leq \mathbf{P}_m(\tau_{m,1}+N_m >n)\vspace{1.5mm}\\
    \hspace{2cm} \leq \mathbf{P}_m(\tau_{m,1} >\epsilon_1 n) + \mathbf{P}_m(N_m >(1-\epsilon_1)n) \vspace{1.5mm}\\ 
    \hspace{2cm} \leq C_1 \e^{-\gamma_1 \epsilon_1 n}+C_2 \e^{-\gamma_2 (1-\epsilon_1)n} \leq \Tilde{C} \e^{-\Tilde{\gamma} n}\vspace{2mm}\\
    \hspace{35mm} \forall n>-(1+\tilde{\epsilon})\log c/I^*(M,K,L),
\ena
\eeq
where $\tilde{\epsilon}=\frac{\epsilon+\epsilon_1}{1-\epsilon_1}$, $\Tilde{C}>0$ and $\Tilde{\gamma}>0$.

For machines $k\neq 1$ in Case $1$ and machines $k$ in Case $2$, for $t\geq\tau_{m,1}$, the detection time can be bounded by summing over all normal cells probed in each machine:
\beq
    \tau^{k}\leq \left(\tau_{m,1} + \sum\limits_{j\neq m} N_{j,1}\right) + \bar{N}^{k} + \sum\limits_{j\in\mathcal{H}^{k}} D_j^{k}.
\eeq
By combining (\ref{eq:Nj1}) with Lemmas \ref{lem:tau_m1}, \ref{lem:barTau_Hk} and \ref{lem:tauj_delay}, for every fixed $\epsilon_1>0$ there exist $\Tilde{C}>0$ and $\Tilde{\gamma}>0$ such that: 
\beq
\label{eq:tau_k}
\bea{l}
    \mathbf{P}_m(\tau^{k} > n) \vspace{1.5mm}\\
    \hspace{1cm} \leq \mathbf{P}_m(\tau_{m,1} >\frac{\epsilon_1}{3}n) + \sum\limits_{j\neq m} \mathbf{P}_m(N_{j,1} >\frac{\epsilon_1}{3(M-1)}n) \vspace{1.5mm}\\
    \hspace{1cm} + \mathbf{P}_m(\bar{N}^{k} >(1-\epsilon_1)n) + \sum\limits_{j\in\mathcal{H}^{k}}\mathbf{P}_m(D_j^{k} >\frac{\epsilon_1}{3}n)\vspace{1.5mm}\\
    \hspace{1cm}\leq \Tilde{C} \e^{-\Tilde{\gamma} n} \hspace{5mm} \forall n>-(1+\tilde{\epsilon})\log c/I^*(M,K,L),
\ena
\eeq
where $\tilde{\epsilon}=\frac{\epsilon+\epsilon_1}{1-\epsilon_1}$, for all $m=1,...,M$.

Now, we can derive the bound on the entire stopping time $\tau$, by using the conclusions of (\ref{eq:max_totalDelay_k}), (\ref{eq:tau_k_1}) and (\ref{eq:tau_k}). For every fixed $\tilde{\epsilon}>0$ there exists $C\triangleq K\cdot\Tilde{C}>0$ and $\Tilde{\gamma}>0$ such that:
\beq
\bea{l}
\label{eq:tau}
    \mathbf{P}_m(\tau > n) \leq \mathbf{P}_m(\max\{\tau^{1},\tau^{2},...,\tau^{k}\} >n) \vspace{1.5mm}\\
    \leq \sum\limits_{k=1}^K \mathbf{P}_m(\tau^{k}>n) 
    \leq K\cdot\Tilde{C} \e^{-\Tilde{\gamma} n}\leq C \e^{-\Tilde{\gamma} n}, \vspace{1.5mm}\\
    \hspace{32mm} \forall n>-(1+\tilde{\epsilon})\log c/I^*(M,K,L).
\ena
\eeq
Finally, applying the tail sum formula for expectation proves the lemma.
\end{proof}\vspace{2mm}

To complete the proof, note that the asymptotic lower bound on the Bayes risk without switching cost (i.e., $s=0$) is given by \cite{cohen2015active}: $ R(\Gamma) \geq \frac{-c\;\log c}{I^*(M,K,L)}$ as $c \rightarrow 0$, which is clearly valid for $s>0$ as well. As a result, applying Statements 1, 2, Lemma \ref{lem:expectedDetectionTime}, and the lower bound on the Bayes risk yields the asymptotic optimality of the CCS policy. \textbf{This proves Statement $3$ of Theorem \ref{th:asymptotic_optimality} for a single anomaly}. \vspace{0.3cm}

The proof of Theorem \ref{th:asymptotic_optimality} for the case of multiple targets $L>1$ follow a similar line of arguments as in the proof for a single anomaly. Hence, we explain here only the slight modifications. First, similar to Lemma \ref{lem:error_upperBound}, it can be verified that declaring the locations of $L$ targets once $\Delta_L S(n) \triangleq S_{m^{(L)}(n)}(n)-S_{m^{(L+1)}(n)}(n)\geq -\log c$, achieves an error probability of $O(c)$. Second, similar to Lemma \ref{lem:tau_m1}, it can be shown that the CCS policy for multiple targets achieves a bounded number of switchings. Third, similar to Lemma \ref{lem:expectedDetectionTime}, it can be verified that the detection time approaches $-\log c/I^*(M,K,L)$. More specifically, all $L$ targets and a fraction $\frac{K-L}{M-L}$ of normal cells are observed at each given time in the asymptotic regime. Therefore, the detection time approaches $\frac{-\log c}{D(g||f)+\frac{K-L}{M-L}D(f||g)}$. \vspace{1mm}
Finally, these three statements yield the asymptotic optimality of the CCS policy for multiple targets, presented in Statement 3 of Theorem \ref{th:asymptotic_optimality}.

\bibliographystyle{ieeetr}

\end{document}

%% file: macros.tex
\usepackage{amsmath}

\newcommand{\e}{\mathop{\mathrm e}\nolimits}    

\makeatletter
\def\BState{\State\hskip-\ALG@thistlm}
\makeatother

\newcounter{theorem}
\newtheorem{theorem}{Theorem}
\newtheorem{lemma}{Lemma}
\newtheorem{claim}[theorem]{Claim}
\newcounter{definition}
\newtheorem{definition}{Definition}
\newcommand{\comment}[1]{}

\newcommand{\beq}{\begin{equation}}
\newcommand{\eeq}{\end{equation}}
\newcommand{\bea}{\begin{array}}
\newcommand{\ena}{\end{array}}
\newcommand{\bds}{\begin {itemize}}
\newcommand{\eds}{\end {itemize}}
\newcommand{\bdf}{\begin{definition}}
\newcommand{\blm}{\begin{lemma}}
\newcommand{\edf}{\end{definition}}
\newcommand{\elm}{\end{lemma}}
\newcommand{\bthm}{\begin{theorem}}
\newcommand{\ethm}{\end{theorem}}
\newcommand{\bprp}{\begin{prop}}
\newcommand{\eprp}{\end{prop}}
\newcommand{\bcl}{\begin{claim}}
\newcommand{\ecl}{\end{claim}}
\newcommand{\bcr}{\begin{coro}}
\newcommand{\ecr}{\end{coro}}
\newcommand{\bquest}{\begin{question}}
\newcommand{\equest}{\end{question}}

\newcommand{\larrow}{{\larrow}}

\newcommand{\nin}{{\not \in}}

\def\urltilda{\kern -.15em\lower .7ex\hbox{\~{}}\kern .04em}